\documentclass[12pt,a4paper]{article}
\usepackage{amsmath,amssymb,amsthm} 
\numberwithin{equation}{section}
\newtheorem{lemma}{Lemma}
\newtheorem{theorem}{Theorem}
\newtheorem{prop}{Proposition}

\theoremstyle{remark}
\newtheorem{remark}{Remark}
\newcommand{\beq}{\begin{equation}}
\newcommand{\eeq}{\end{equation}}
\newcommand{\beqnn}{\begin{equation*}}
\newcommand{\eeqnn}{\end{equation*}}
\newcommand{\rd}{\partial}

\newcommand{\tp}[1]{\:{}^{\mathrm{t}}#1}
\newcommand{\CC}{\mathbb{C}}
\newcommand{\PP}{\mathbb{P}}

\newcommand{\ZZ}{\mathbb{Z}}
\newcommand{\GL}{\mathrm{GL}}
\newcommand{\gl}{\mathrm{gl}}
\newcommand{\bst}{\boldsymbol{t}}
\newcommand{\bsx}{\boldsymbol{x}}
\newcommand{\bsT}{\boldsymbol{T}}
\newcommand{\bszero}{\boldsymbol{0}}

\newcommand{\calO}{\mathcal{O}}
\newcommand{\calP}{\mathcal{P}}

\newcommand{\frakL}{\mathfrak{L}}
\newcommand{\bstbar}{\bar{\bst}}
\newcommand{\tbar}{\bar{t}}
\newcommand{\ubar}{\bar{u}}
\newcommand{\wbar}{\bar{w}}
\newcommand{\Bbar}{\bar{B}}
\newcommand{\Hbar}{\bar{H}}
\newcommand{\Tbar}{\bar{T}}
\newcommand{\Wbar}{\bar{W}}
\newcommand{\Lbar}{\bar{L}}
\newcommand{\bsTbar}{\bar{\bsT}}

\begin{document}

\title{\bf Orbifold melting crystal models\\ 
and reductions of Toda hierarchy}
\author{Kanehisa Takasaki\thanks{E-mail: takasaki@math.h.kyoto-u.ac.jp}\\
{\normalsize Department of Mathematics, Kinki University}\\ 
{\normalsize 3-4-1 Kowakae, Higashi-Osaka, Osaka 577-8502, Japan}}
\date{}
\maketitle 

\begin{abstract}
Orbifold generalizations of the ordinary and modified 
melting crystal models are introduced.  They are labelled 
by a pair $a,b$ of positive integers, and geometrically related 
to $\ZZ_a\times\ZZ_b$ orbifolds of local $\CC\PP^1$ geometry 
of the $\calO(0)\oplus\calO(-2)$ and $\calO(-1)\oplus\calO(-1)$ types. 
The partition functions have a fermionic expression in terms 
of charged free fermions. With the aid of shift symmetries 
in a fermionic realization of the quantum torus algebra, 
one can convert these partition functions to tau functions 
of the 2D Toda hierarchy.  The powers $L^a,\Lbar^{-b}$ 
of the associated Lax operators turn out to take a special 
factorized form that defines a reduction of the 2D Toda hierarchy.  
The reduced integrable hierarchy for the orbifold version 
of the ordinary melting crystal model is the bi-graded Toda hierarchy 
of bi-degree $(a,b)$. That of the orbifold version 
of the modified melting crystal model is the rational reduction 
of bi-degree $(a,b)$.  This result seems to be in accord with 
recent work of Brini et al. on a mirror description 
of the genus-zero Gromov-Witten theory on a $\ZZ_a\times\ZZ_b$ 
orbifold of the resolved conifold. 
\end{abstract}

\begin{flushleft}
2010 Mathematics Subject Classification: 
17B65, 35Q55, 81T30, 82B20\\
Key words: melting crystal, orbifold model, free fermion, 
quantum torus, shift symmetry, Toda hierarchy, 
rational reduction
\end{flushleft}

\newpage 

\section{Introduction}

Recently, we extended our previous work \cite{Takasaki12,Takasaki13} 
on the integrable structure of a modified melting crystal model 
to open string amplitudes of topological string theory 
on a generalized conifold \cite{Takasaki14}. 
The modified meting crystal model is essentially 
a reformulation of open string amplitudes, 
or local Gromov-Witten invariants, on the resolved conifold.  
Brini \cite{Brini10} conjectured, and partially proved at low orders 
of genus expansion, a mirror-theoretic correspondence 
between the local Gromov-Witten theory of the resolved conifold 
and the Ablowitz-Ladik hierarchy \cite{AL75,Vekslerchik97}. 
In the course of refining this observation, 
Brini and his collaborators \cite{BCR11} reformulated 
this integrable hierarchy as a special ``rational reduction'' 
of the 2D Toda hierarchy \cite{UT84,TT95}.  
We considered this issue from a different route 
once developed for the ordinary melting crystal model 
\cite{NT07,NT08}.  Technical clues therein are --- 
\begin{itemize}
\item[(i)] a fermionic expression of the partition function, 
\item[(ii)] an associated fermionic realization 
of the quantum torus algebra, 
\item[(iii)] algebraic relations called ``shift symmetries'' 
in this algebra, 
\item[(iv)] a matrix factorization problem that solves 
the initial value problem in the Lax formalism. 
\end{itemize}
Armed with these tools, we proved that 
the Ablowitz-Ladik hierarchy indeed underlies 
the modified melting crystal model \cite{Takasaki12,Takasaki13}, 
and extended this result to a generalized conifold  
\cite{Takasaki14}.  This generalized conifold 
is of the {\it bubble} type, namely, its web diagram 
is a linear chain of repetition of the web diagram 
of the resolved conifold.  The relevant integrable hierarchy 
in this case is a kind of generalized Ablowitz-Ladik hierarchy 
that is realized as a reduction of the 2D Toda hierarchy. 

In this paper, we present yet another extension of our study 
on the melting crystal models.  This work is motivated 
by a very recent paper of Brini et al. \cite{BCRR14} 
on a general scheme of rational reductions 
of the 2D Toda hierarchy and its application 
to the mirror theory of a toric Calabi-Yau {\it orbifold}.  
They proposed a large class of rational reductions, 
and proved that the dispersionless limit 
of a particular reduction captures 
the genus-zero Gromov-Witten theory 
of a $\ZZ_a\times\ZZ_b$ orbifold of the resolved conifold.  
Remarkably, the reduced integrable hierarchy therein, 
too, is a generalized Ablowitz-Ladik hierarchy, 
but totally different from ours.  
In the case of the orbifold, the $a$-th and $b$-th powers 
of the Lax operators $L,\Lbar^{-1}$ of the 2D Toda hierarchy 
are expected to take the ``rational'' form 
\beqnn
  L^a = BC^{-1},\quad \Lbar^{-b} = DCB^{-1}, 
\eeqnn
where $D$ is a non-zero constant, normalized to $1$ 
in the setting of Brini et al., and $B$ and $C$ 
are difference operators of the form
\beqnn
\begin{gathered}
  B = e^{a\rd_s} + \beta_1(s)e^{(a-1)\rd_s} + \cdots + \beta_a(s),\\
  C = 1 + \gamma_1(s)e^{-\rd_s} + \cdots + \gamma_b(s)e^{-b\rd_s} 
\end{gathered}
\eeqnn
on the 1D lattice with coordinate $s$.  
In the case of the bubble type \cite{Takasaki14}, 
the Lax operators themselves take the rational form 
\beqnn
  L = Be^{(1-N)\rd_s}C^{-1},\quad 
  \Lbar^{-1} = DCe^{(N-1)\rd_s}B^{-1}, 
\eeqnn
where $N$ is the number of repetition 
of the conifold diagram, $D$ is a non-zero constant, 
and $B$ and $C$ are difference operators of the form 
\beqnn
\begin{gathered}
  B = e^{N\rd_s} + \beta_1(s)e^{(N-1)\rd_s} + \cdots + \beta_N(s),\\
  C = 1 + \gamma_1(s)e^{-\rd_s} + \cdots + \gamma_N(s)e^{-N\rd_s}. 
\end{gathered}
\eeqnn
One of our goals is to derive the former 
--- ``the rational reduction of bi-degree $(a,b)$'' 
in the terminology of Brini et al. \cite{BCRR14} --- 
from a generalized melting crystal model.  

Such a melting crystal model can be found in the work 
of Bryan et al. \cite{BCY10} on an orbifold version  
of the method of topological vertex \cite{AKMV03}.  
They illustrated their method with two examples. 
One of them is local $\CC\PP^1_{a,b}$ geometry, namely, 
the total space of the bundle $\calO(-1)\oplus\calO(-1)$ 
over a $\ZZ_a\times\ZZ_b$ orbifold $\CC\PP^1_{a,b}$ of $\CC\PP^1$. 
This, too, is an orbifold generalization of the resolved conifold.  
The orbifold topological vertex construction 
gives a generating function of Donaldson-Thomas invariants 
of this Calabi-Yau threefold.  Its main part may be 
thought of as a generalized melting crystal model.  
We call this model {\it an orbifold melting crystal model\/}. 
In the non-orbifold case where $a = b = 1$, 
this model reduces to the modified melting crystal model.  
We shall show that this orbifold version 
of the modified melting crystal model is indeed 
accompanied by the rational reduction of bi-degree $(a,b)$ 
of the 2D Toda hierarchy.  

In the same sense, one can construct an orbifold version 
of the ordinary melting crystal model.  Geometrically, 
this amounts to local $\CC\PP^1_{a,b}$ geometry 
of the $\calO(0)\oplus\calO(-2)$ type.  We shall show, 
as another goal, that this orbifold model corresponds 
to the bi-graded Toda hierarchy of bi-degree $(a,b)$ 
\cite{Kuperschmidt85,Carlet06}.  
This integrable hierarchy is a reduction 
of the 2D Toda hierarchy for which the Lax operators 
satisfy the reduction condition 
\beqnn
  L^a = D^{-1}\Lbar^{-b}, 
\eeqnn
where $D$ is a non-zero constant that is usually normalized 
to $1$.  Both sides of this condition become  
a difference operator of the form 
\beqnn
  \frakL = e^{a\rd_s} + \alpha_1(s)e^{(a-1)\rd_s} + \cdots 
           + \alpha_{a+b}(s)e^{-b\rd_s}. 
\eeqnn
When $a = b = 1$, $\frakL$ reduces to the well known 
Lax operator of the 1D Toda hierarchy.  
Just like the 1D Toda hierarchy \cite{CDZ04,Milanov05,Takasaki10}, 
the bi-graded Toda hierarchy can be extended 
by extra logarithmic flows \cite{Carlet06}, 
and has been applied to the Gromov-Witten theory 
of an orbifold of $\CC\PP^1$ \cite{MT06,CvdL13}. 
In the case of the orbifold melting crystal model, 
both sides of the reduction condition are further factorized as 
\beqnn
  L^a = D^{-1}\Lbar^{-b} = BC, 
\eeqnn
where $B$ and $C$ are difference operators of the same form 
as those in the rational reduction of bi-degree $(a,b)$.  
Our previous melting crystal model with ``two $q$-parameters'' 
\cite{Takasaki09} 
\footnote{Note that this model is not directly related 
to the so called ``refinement''. }
turns out to be a special case of this orbifold model.  

The most prominent characteristic of the orbifold cases 
is that the powers $L^a,\Lbar^{-b}$ of $L,\Lbar^{-1}$ 
rather than $L,\Lbar^{-1}$ themselves show up 
in the reduction condition.  This affects 
the correspondence between the coupling constants 
$t_k,\tbar_k$ of the models and the time variables 
of the 2D Toda hierarchy as well.  
From a technical point of view, this characteristic 
is related to {\it fractional framing factors\/}  
in the orbifold topological vertex construction \cite{BCY10}. 
We shall encounter their avatars in our calculations. 
They are responsible for the emergence of the powers 
$L^a,\Lbar^{-b}$ of the Lax operators.  

This paper is organized as follows.  
The orbifold version of the ordinary model 
is referred to as {\it the first orbifold model\/}. 
The orbifold version of the modified model 
is distinguished as {\it the second orbifold model\/}.  
Throughout the paper, these models are treated 
in a fully parallel way.  Section 2 is devoted 
to formulation of these models. 
Sections 3 and 4 are focused on implications 
of the fermionic expression of the partition functions. 
Section 5 and 6 present the aforementioned results 
on the Lax operators.   Let us explain the contents 
in more detail. 

In Section 2, these two orbifold models are defined 
as models of random partitions, and translated 
to the language of charged free fermions.  
The Boltzmann weights of the partition functions 
are built from special values of the infinite-variable 
Schur functions $s_\lambda(\bsx)$, $\bsx = (x_1,x_2,\ldots)$.  
These weights are deformed by external potentials 
$\Phi_k(\lambda,s)$, where $k = 1,2,\ldots$ 
in the first orbifold model and $k = \pm 1,\pm 2,\ldots$. 
in the second orbifold model.  
Note that these potentials depend 
on the lattice coordinate $s$ as well. 
The partition function $Z_{a,b}(s,\bst)$ 
of the first orbifold model is a function of $s$ 
and a set of coupling constants $\bst = (t_1,t_2,\ldots)$.  
The partition function $Z'_{a,b}(s,\bst,\bstbar)$ 
of the second orbifold model depends on $s$ and two sets 
of coupling constants $\bst = (t_1,t_2,\ldots)$ 
and $\bstbar  = (\tbar_1,\tbar_2,\ldots)$. 

In Sections 3 and 4, these partition functions 
are converted to tau functions of the 2D Toda hierarchy. 
This is a somewhat lengthy and complicated generalization 
of the calculations for the ordinary \cite{NT07,NT08} and 
modified \cite{Takasaki12,Takasaki13} melting crystal models.   
We start from the fermionic expression 
of the partition functions  that does not look like 
tau functions. This expression contains the generators 
$e^{H(\bst)}$ and $e^{\Hbar(\bstbar)}$ of time evolutions.  
$H(\bst)$ and $\Hbar(\bstbar)$ are linear combinations 
of fermion bilinears $H_k$ and $H_{-k}$ that are elements 
of a fermionic realization of the quantum torus algebra. 
We use the aforementioned shift symmetries  
of this algebra to transform $H_{\pm k}$'s 
to elements of the $U(1)$ current algebra, 
namely, infinitesimal generators of time evolutions 
of the 2D Toda hierarchy.   By these calculations, 
the partition functions turn out to be, 
up to a relatively simple prefactor, tau functions 
of the 2D Toda hierarchy (Theorems 1 and 2).  

In Sections 5 and 6, these special solutions 
of the 2D Toda hierarchy are re-interpreted 
in the Lax formalism.  To this end, we make use 
of a matrix factorization problem that determines 
the dressing operators $W,\Wbar$ and 
the Lax operators $L,\Lbar$.  (These operators are 
translated to $\ZZ\times\ZZ$ matrices in advance.) 
Fortunately, the factorization problem 
in the present setting can be solved explicitly 
at the initial time $\bst = \bstbar = \bszero$. 
This enables us, after some algebra, to derive 
the factorized form of $L^a$ and $\Lbar^{-b}$ 
at the initial time.  Since the factorized form 
is known to be preserved by all time evolutions 
of the 2D Toda hierarchy, this is enough 
to deduce the conclusion (Theorem 3).

\section{Orbifold melting crystal models}

\subsection{Partition functions}

The ordinary melting crystal model \cite{ORV03,MNTT04} 
and its resolved conifold version \cite{BP04,CGMPS06} 
are defined, respectively, by the partition functions 
\begin{gather}
  Z = \sum_{\lambda\in\calP}s_\lambda(q^{-\rho})^2Q^{|\lambda|},
    \label{Z}\\
  Z' = \sum_{\lambda\in\calP}
       s_\lambda(q^{-\rho})s_{\tp{\lambda}}(q^{-\rho})Q^{|\lambda|}, 
    \label{Z'}
\end{gather}
where $\calP$ denotes the set of all partitions 
$\lambda = (\lambda_i)_{i=1}^\infty$, $\tp{\lambda}$ 
the conjugate (or transposed) partition of $\lambda$, 
$|\lambda|$ the sum of all parts $\lambda_i$, 
and $s_\lambda(q^{\-\rho})$ the special value of 
the infinite-variable Schur function $s_\lambda(\bsx)$, 
$\bsx = (x_1,x_2,\ldots)$ \cite{Macdonald-book}, at 
\beqnn
  q^{-\rho} = \left(q^{1/2},q^{3/2},\ldots,q^{n-1/2},\ldots\right). 
\eeqnn

The orbifold models are obtained by replacing 
the Boltzmann weights as 
\beqnn
  s_\lambda(q^{-\rho})^2 \to 
  s_\lambda(p_1q^{-\rho},\ldots,p_aq^{-\rho})
  s_\lambda(r_1q^{-\rho},\ldots,r_bq^{-\rho})
\eeqnn
and 
\beqnn
  s_\lambda(q^{-\rho})s_{\tp{\lambda}}(q^{-\rho}) \to 
  s_\lambda(p_1q^{-\rho},\ldots,p_aq^{-\rho})
  s_{\tp{\lambda}}(r_1q^{-\rho},\ldots,r_bq^{-\rho}) 
\eeqnn
in the foregoing partition functions $Z$ and $Z'$.  
Here $p_1,p_2,\ldots,p_a$ and $r_1,r_2,\ldots,p_b$ 
are two sets of parameters of these models, and 
$(p_1q^{-\rho},\ldots,p_aq^{-\rho})$ and 
$(r_1q^{-\rho},\ldots,r_bq^{-\rho})$ stand for the union of 
\beqnn
  p_iq^{-\rho} 
  = \left(p_iq^{1/2},p_iq^{3/2},\ldots,p_iq^{n-1/2},\ldots\right),\quad
  i = 1,\ldots,a
\eeqnn
and 
\beqnn
  r_jq^{-\rho}
  = \left(r_jq^{1/2},r_jq^{3/2},\ldots,r_jq^{n-1/2},\ldots\right),\quad 
  j = 1,\ldots,b, 
\eeqnn
respectively. Since the Schur functions are symmetric functions, 
one can do such regrouping and reordering of arguments freely. 
The partition functions of the orbifold models 
are sums of these Boltzmann weights over all partitions: 
\begin{gather}
  Z_{a,b} = \sum_{\lambda\in\calP}
    s_\lambda(p_1q^{-\rho},\ldots,p_aq^{-\rho})
    s_\lambda(r_1q^{-\rho},\ldots,r_bq^{-\rho})Q^{|\lambda|}, 
    \label{Zab}\\
  Z'_{a,b} = \sum_{\lambda\in\calP}
    s_\lambda(p_1q^{-\rho},\ldots,p_aq^{-\rho})
    s_{\tp{\lambda}}(r_1q^{-\rho},\ldots,r_bq^{-\rho})Q^{|\lambda|}. 
    \label{Z'ab}
\end{gather}
$Z'_{a,b}$ is essentially the Donaldson-Thomas 
partition function of local $\CC\PP^1_{a,b}$ geometry 
presented by Bryan et al. \cite{BCY10}.  

By the Cauchy identities 
\beq
\begin{gathered}
  \sum_{\lambda\in\calP}
    s_\lambda(x_1,x_2,\ldots)s_\lambda(y_1,y_2,\ldots)
  = \prod_{m,n=1}^\infty (1 - x_my_n)^{-1},\\
  \sum_{\lambda\in\calP}
    s_{\lambda}(x_1,x_2,\ldots)s_{\tp{\lambda}}(y_1,y_2,\ldots)
  = \prod_{m,n=1}^\infty (1 + x_my_n),
\end{gathered}
\eeq
the partition functions turn into the product form 
\begin{gather}
  Z_{a,b} = \prod_{i=1}^a\prod_{j=1}^b M(p_ir_jQ,q),
    \label{Zab-Mac}\\
  Z'_{a,b} = \prod_{i=1}^a\prod_{j=1}^b M(-p_ir_jQ,q)^{-1},
    \label{Z'ab-Mac} 
\end{gather}
where $M(x,q)$ denotes the MacMahon function 
\beqnn
  M(x,q) =\prod_{n=1}^\infty(1 - xq^n)^{-n}.
\eeqnn

\begin{remark}
Since homogeneity of the Schur functions imply the identities 
\beq
\begin{gathered}
  s_\lambda(p_1q^{-\rho},\ldots,p_aq^{-\rho}) 
  = s_\lambda\left(\frac{p_1}{p_a}q^{-\rho},\ldots,
      \frac{p_{a-1}}{p_a}q^{-\rho},q^{-\rho}\right)p_a^{|\lambda|},\\
  s_\lambda(r_1q^{-\rho},\ldots,r_bq^{-\rho}) 
  = s_\lambda\left(\frac{r_1}{r_b}q^{-\rho},\ldots,
      \frac{r_{b-1}}{r_b}q^{-\rho},q^{-\rho}\right)r_b^{|\lambda|}, 
\end{gathered}
\label{normalize-pr1}
\eeq
one can normalize the parameters of the models as 
\beq
  p_a = r_b = 1 
\label{normalize-pr2}
\eeq
by replacing $Q \to Q(p_ar_b)^{-1}$.  
\end{remark}

\begin{remark}
\label{remark-2q}
Since 
\beqnn
\begin{gathered}
  s_\lambda(q^{-\rho/a}) 
  = s_\lambda(q^{(1-a)/2a}q^{-\rho},q^{(3-a)/2a}q^{-\rho},
      \ldots,q^{(a-1)/2a}q^{-\rho}),\\
  s_\lambda(q^{-\rho/b}) 
  = s_\lambda(q^{(1-b)/2b}q^{-\rho},q^{(3-b)/2b}q^{-\rho},
      \ldots,q^{(b-1)/2b}q^{-\rho}),
\end{gathered}
\eeqnn
these orbifold models reduce to the modifications 
\begin{gather}
  \tilde{Z}_{a,b} = \sum_{\lambda\in\calP}
    s_\lambda(q^{\rho/a})s_\lambda(q^{\rho/b})Q^{|\lambda|},
    \label{Ztilde-ab}\\
  \tilde{Z'}_{a,b} = \sum_{\lambda\in\calP}
    s_\lambda(q^{\rho/a})s_{\tp{\lambda}}(q^{\rho/b})Q^{|\lambda|}
    \label{Z'tilde-ab}
\end{gather}
of the ordinary models (\ref{Z}) and (\ref{Z'}) 
when the parameters are specialized as 
\beqnn
\begin{gathered}
  p_1 = q^{(1-a)/2a},\; p_2 = q^{(3-a)/2a},\;\ldots\; p_a = q^{(a-1)/2a},\\
  r_1 = q^{(1-b)/2b},\; r_2 = q^{(3-b)/2b},\;\ldots\; r_b = q^{(b-1)/2}. 
\end{gathered}
\eeqnn
This is an unexpected link with our previous work 
on a melting crystal model with two $q$-parameters 
\cite{Takasaki09}.  The previous model coincides 
with $\tilde{Z}_{a,b}$ when the $q$-parameters $q_1,q_2$ 
therein are specialized as 
\beq
  q_1 = q^{1/a},\quad q_2 = q^{1/b}. 
\label{q1q2-special}
\eeq

\end{remark}

\subsection{Fermionic formulation}

The setup for the fermionic formulation \cite{MJD-book} 
is the same as our earlier work \cite{NT07,NT08}: 
\begin{itemize}
\item[(i)] Fourier modes $\psi_n,\psi^*_n$ of the 2D free fermion 
fields $\psi(z),\psi^*(z)$ are labelled by integers $n\in\ZZ$, 
and satisfy the anti-commutation relations 
\beqnn
  \psi_m\psi^*_n + \psi^*_n\psi_m = \delta_{m+n,0}, \quad
  \psi_m\psi_n + \psi_n\psi_m = 0, \quad 
   \psi^*_m\psi^*_n + \psi^*_n\psi^*_m = 0. 
\eeqnn
\item[(ii)] The Fock space and its dual space are decomposed 
to the charge-$s$ sectors for $s \in \ZZ$.  An orthonormal basis 
of the charge-$s$ sector is given by the ground states 
\beqnn
\begin{gathered}
  \langle s| = \langle -\infty|\cdots\psi^*_{s-1}\psi^*_s,\\
  |s\rangle = \psi_{-s}\psi_{-s+1}\cdots|-\infty\rangle 
\end{gathered}
\eeqnn
and the excited states 
\beqnn
\begin{gathered}
  \langle s,\lambda| 
  = \langle -\infty|\cdots\psi^*_{\lambda_i+s-i+1}
    \cdots\psi^*_{\lambda_2+s-1}\psi^*_{\lambda_1+s},\\
  |s,\lambda\rangle 
  = \psi_{-\lambda_1-s}\psi_{-\lambda_2-s+1}\cdots
    \psi_{-\lambda_i-s+i-1}\cdots|-\infty\rangle 
\end{gathered}
\eeqnn
labelled by partitions $\lambda = (\lambda_i)_{i=1}^\infty \in \calP$. 
\item[(iii)] The action of the fermion bilinears 
\beqnn
\begin{gathered}
  L_0 = \sum_{n\in\ZZ}n{:}\psi_{-n}\psi^*_n{:},\quad
  W_0 = \sum_{n\in\ZZ}n^2{:}\psi_{-n}\psi^*_n{:},\\
  H_k = \sum_{n\in\ZZ}q^{kn}{:}\psi_{-n}\psi^*_n{:},\quad
  J_k = \sum_{n\in\ZZ}{:}\psi_{-n}\psi^*_{n+k}{:} 
\end{gathered}
\eeqnn
on the Fock space, where $:\quad:$ denotes 
the normal ordering with respect to 
the vacuum states $\langle 0|$, $|0\rangle$, 
preserves the charge-$s$ sector.  
$J_0,L_0,W_0$ and $H_k$'s are diagonal with respect to 
the basis $\{|\lambda,s\rangle\}_{\lambda\in\calP}$: 
\begin{align}
  \langle\lambda,s|J_0|\mu,s\rangle 
  &= \delta_{\lambda\mu}s,\\
  \langle\lambda,s|L_0|\mu,s\rangle 
  &= \delta_{\lambda\mu}\left(|\lambda|+\frac{s(s+1)}{2}\right), \\
  \langle\lambda,s|W_0|\mu,s\rangle
  &= \delta_{\lambda\mu}\left(\kappa(\lambda) + (2s+1)|\lambda| 
     + \frac{s(s+1)(2s+1)}{6}\right), \\
  \langle\lambda,s|H_k|\mu,s\rangle 
  &= \delta_{\lambda\mu}\Phi_k(\lambda,s), 
\end{align}
where 
\beqnn
  \kappa(\lambda) 
  = \sum_{i=1}^\infty 
    \left(\left(\lambda_i-i+\frac{1}{2}\right)^2 
        - \left(-i+\frac{1}{2}\right)^2\right)
\eeqnn
and 
\beqnn
  \Phi_k(\lambda,s) 
  = \sum_{i=1}^\infty(q^{k(\lambda_i+s-i+1)} - q^{k(s-i+1)}) 
   + \frac{1-q^{ks}}{1-q^k}q^k. 
\eeqnn
\item[(vi)] The vertex operators \cite{OR01,BY08} 
\beqnn
  \Gamma_{\pm}(z) 
  = \exp\left(\sum_{k=1}^\infty\frac{z^k}{k}J_{\pm k}\right),\quad
  \Gamma'_{\pm}(z) 
  = \exp\left(- \sum_{k=1}^\infty\frac{(-z)^k}{k}J_{\pm k}\right) 
\eeqnn
and the multi-variable extensions 
\beqnn
  \Gamma_{\pm}(\bsx) = \prod_{i\ge 1}\Gamma_{\pm}(x_i),\quad 
  \Gamma'_{\pm}(\bsx) = \prod_{i\ge 1}\Gamma'_{\pm}(x_i) 
\eeqnn
to $\bsx = (x_1,x_2,\ldots)$, too, preserve the charge-$s$ sector. 
The matrix elements become skew-Schur functions 
\cite{Macdonald-book,MJD-book}
\beq
\begin{gathered}
  \langle\lambda,s|\Gamma_{-}(\bsx)|\mu,s\rangle
  = \langle\mu,s|\Gamma_{+}(\bsx)|\lambda,s\rangle
  = s_{\lambda/\mu}(\bsx),\\
  \langle\lambda,s|\Gamma'_{-}(\bsx)|\mu,s\rangle 
  = \langle\mu,s|\Gamma'_{+}(\bsx)|\lambda,s\rangle
  = s_{\tp{\lambda}/\tp{\mu}}(\bsx). 
\end{gathered}
\eeq
\end{itemize}

A fermionic expression of $Z_{a,b}$ can be derived as follows. 
Note that the special values of the Schur functions 
in the Boltzmann weights of $Z_{a,b}$ can be expressed 
in a fermionic form as 
\beqnn
\begin{gathered}
  s_\lambda(p_1q^{-\rho},\cdots,p_aq^{-\rho}) 
  = \langle 0|\Gamma_{+}(p_1q^{-\rho})
    \cdots\Gamma_{+}(p_aq^{-\rho})|\lambda,0\rangle, \\
  s_\lambda(r_1q^{-\rho},\cdots,r_bq^{-\rho}) 
  = \langle\lambda,0|\Gamma_{-}(r_bq^{-\rho})
    \cdots\Gamma_{-}(r_1q^{-\rho})|0\rangle. 
\end{gathered}    
\eeqnn
Normalize $p_a$ and $r_b$ as shown in (\ref{normalize-pr1}) 
and (\ref{normalize-pr2}), and introduce 
new parameters $P_1,P_2,\ldots,P_{a-1}$ 
and $R_1,R_2,\ldots,R_{b-1}$ as 
\beq
\begin{gathered}
  p_i = P_i\cdots P_{a-1}\quad\text{for $i = 1,2,\ldots,a-1$},\quad 
  p_a = 1,\\
  r_j = R_j\cdots R_{b-1}\quad\text{for $j = 1,2,\ldots,b-1$},\quad 
  r_b = 1. 
\end{gathered}
\label{pr-PR}
\eeq
Recall that setting $p_a = r_b = 1$ does not lead 
to loss of generality.  The products of $\Gamma_{\pm}$'s 
in the expression of the Schur functions thereby become 
\begin{align*}
  &\Gamma_{+}(p_1q^{-\rho})\cdots\Gamma_{+}(p_{a-1}q^{-\rho})
   \Gamma_{+}(q^{-\rho})\nonumber\\
  &= \Gamma_{+}(P_1\cdots P_{a-1}q^{-\rho})\cdots
     \Gamma_{+}(P_{a-1}q^{-\rho})\Gamma_{+}(q^{-\rho})\nonumber\\
  &= (P_1\cdots P_{a-1})^{-L_0}\Gamma_{+}(q^{-\rho})
     P_1^{L_0}\Gamma_{+}(q^{-\rho})P_2^{L_0}\cdots
     \Gamma_{+}(q^{-\rho})P_{a-1}^{L_0}\Gamma_{+}(q^{-\rho}) 
\end{align*}
and 
\begin{align*}
  &\Gamma_{-}(q^{-\rho})\Gamma_{-}(r_{b-1}q^{-\rho})
   \cdots\Gamma_{-}(r_1q^{-\rho})\nonumber\\
  &= \Gamma_{-}(q^{-\rho})\Gamma_{-}(R_{b-1}q^{-\rho})\cdots
     \Gamma_{-}(R_{b-1}\cdots R_2q^{-\rho})
     \Gamma_{-}(R_{b-1}\cdots R_1q^{-\rho})\nonumber\\
  &= \Gamma_{-}(q^{-\rho})R_{b-1}^{L_0}\Gamma_{-}(q^{-\rho})
     \cdots R_2^{L_0}\Gamma_{-}(q^{-\rho})R_1^{L_0}
     \Gamma_{-}(q^{-\rho})(R_1\cdots R_{b-1})^{-L_0}. 
\end{align*}
The scaling property 
\beq
  u^{L_0}J_ku^{-L_0} = u^{-k}J_k, \quad k\in\ZZ, 
\label{J-scaling}
\eeq
of $J_k$'s and its consequences 
\beq
  \Gamma_{+}(vq^{-\rho})u^{L_0} = u^{L_0}\Gamma_{+}(uvq^{-\rho}),\quad 
  u^{L_0}\Gamma_{-}(vq^{-\rho}) = \Gamma_{-}(uvq^{-\rho})u^{L_0}
\label{Gamma-scaling}
\eeq
have been used here.  Since the leftmost factor 
$(P_1\cdots P_{a-1})^{-L_0}$ and the rightmost factor 
$(R_1\cdots R_{b-1})^{-L_0}$ disappear as they hit 
the vacuum vector, the fermionic expression 
of the special values of the Schur functions 
can be rewritten as 
\begin{multline*}
  s_\lambda(p_1q^{-\rho},\ldots,p_{a-1}q^{-\rho},q^{-\rho}) \\
  = \langle 0|\Gamma_{+}(q^{-\rho})P_1^{L_0}\Gamma_{+}(q^{-\rho})P_2^{L_0}\cdots
    \Gamma_{+}(q^{-\rho})P_{a-1}^{L_0}\Gamma_{+}(q^{-\rho})|\lambda,0\rangle
\end{multline*}
and 
\begin{multline*}
  s_\lambda(r_1q^{-\rho},\ldots,r_{b-1}q^{-\rho},q^{-\rho}) \\
  = \langle\lambda,0|\Gamma_{-}(q^{-\rho})R_{b-1}^{L_0}\Gamma_{-}(q^{-\rho})
    \cdots R_2^{L_0}\Gamma_{-}(q^{-\rho})R_1^{L_0}\Gamma_{-}(q^{-\rho})|0\rangle. 
\end{multline*}
Thus the partition function $Z_{a,b}$ can be cast into 
the fermionic expression 
\begin{multline}
  Z_{a,b} = \langle 0|\Gamma_{+}(q^{-\rho})P_1^{L_0}
            \Gamma_{+}(q^{-\rho})P_2^{L_0}\cdots 
            \Gamma_{+}(q^{-\rho})P_{a-1}^{L_0}\Gamma_{+}(q^{-\rho})Q^{L_0}\\
  \mbox{}\times 
            \Gamma_{-}(q^{-\rho})R_{b-1}^{L_0}\Gamma_{-}(q^{-\rho})
            \cdots R_2^{L_0}\Gamma_{-}(q^{-\rho})
            P_1^{L_0}\Gamma_{-}(q^{-\rho})|0\rangle.
\label{Zab-fermion}
\end{multline}

In the same way, using the primed versions 
\begin{multline*}
  s_{\tp{\lambda}}(r_1q^{-\rho},\ldots,r_{b-1}q^{-\rho},q^{-\rho}) \\
  = \langle\lambda,0|\Gamma'_{-}(q^{-\rho})R_{b-1}^{L_0}\Gamma'_{-}(q^{-\rho})
    \cdots R_2^{L_0}\Gamma'_{-}(q^{-\rho})R_1^{L_0}\Gamma'_{-}(q^{-\rho})|0\rangle 
\end{multline*}
and 
\beq
  \Gamma'_{+}(vq^{-\rho})u^{L_0} = u^{L_0}\Gamma'_{+}(uvq^{-\rho}),\quad
  u^{L_0}\Gamma'_{-}(vq^{-\rho}) = \Gamma'_{-}(vq^{-\rho})u^{L_0}
\label{Gamma'-scaling}
\eeq
of the foregoing relations as well, one can derive 
the following fermionic expression of $Z'_{a,b}$: 
\begin{multline}
  Z'_{a,b} = \langle 0|\Gamma_{+}(q^{-\rho})P_1^{L_0}
            \Gamma_{+}(q^{-\rho})P_2^{L_0}\cdots 
            \Gamma_{+}(q^{-\rho})P_{a-1}^{L_0}\Gamma_{+}(q^{-\rho})Q^{L_0}\\
  \mbox{}\times 
            \Gamma'_{-}(q^{-\rho})R_{b-1}^{L_0}\Gamma'_{-}(q^{-\rho})
            \cdots R_2^{L_0}\Gamma'_{-}(q^{-\rho})
            P_1^{L_0}\Gamma'_{-}(q^{-\rho})|0\rangle.
\label{Z'ab-fermion}
\end{multline}
Note that the $\Gamma_{-}$'s in the expression of $Z_{a,b}$ 
are replaced by $\Gamma'_{-}$'s.

\subsection{Deformations by external potentials}

The deformed partition functions $Z_{a,b}(s,\bst)$ 
and $Z'_{a,b}(s,\bst,\bstbar)$, 
$\bst = (t_1,t_2,\ldots)$, 
$\bstbar = (\tbar_1,\tbar_2,\ldots)$, 
are defined as 
\begin{multline}
  Z_{a,b}(s,\bst)
  = \sum_{\lambda\in\calP}
    s_\lambda(p_1q^{-\rho},\ldots,p_aq^{-\rho})
    s_\lambda(r_1q^{-\rho},\ldots,r_bq^{-\rho})\\
  \mbox{}\times
    Q^{|\lambda|+s(s+1)/2}e^{\Phi(\lambda,s,\bst)} 
\label{Zab(s,t)}
\end{multline}
and 
\begin{multline}
  Z'_{a,b}(s,\bst,\bstbar)
  = \sum_{\lambda\in\calP}
    s_\lambda(p_1q^{-\rho},\ldots,p_aq^{-\rho})
    s_{\tp{\lambda}}(r_1q^{-\rho},\ldots,r_bq^{-\rho})\\
  \mbox{}\times 
    Q^{|\lambda|+s(s+1)/2}e^{\Phi(\lambda,s,\bst,\bstbar)},
\label{Z'ab(s,t,tbar)}
\end{multline}
where 
\beqnn
\begin{gathered}
  \Phi(\lambda,s,\bst) = \sum_{k=1}^\infty t_k\Phi_k(\lambda,s),\\
  \Phi(\lambda,s,\bst,\bstbar) 
    = \sum_{k=1}^\infty t_k\Phi_k(\lambda,s) 
      + \sum_{k=1}^\infty\tbar_k\Phi_{-k}(\lambda,s).
\end{gathered}
\eeqnn
Under the normalization of parameters as shown in (\ref{pr-PR}), 
\footnote{To absorb the extra factors $p_a^{|\lambda|}$ 
and $r_b^{|\lambda|}$ in (\ref{normalize-pr1}) 
into redefinition of $Q$ in this $s$-dependent setting, 
we have to correct these partition functions 
by the overall factor $(p_ar_b)^{s(s+1)/2}$.  
$Z_{a,b}(s,\bst)$ and $Z'_{a,b}(s,\bst,\bstbar)$ 
in the fermionic expression stand for those corrected 
partition functions.}
these partition functions 
can be converted to a fermionic form as 
\begin{multline}
  Z_{a,b}(s,\bst) 
    = \langle s|\Gamma_{+}(q^{-\rho})P_1^{L_0}\Gamma_{+}(q^{-\rho})P_2^{L_0}\cdots 
      \Gamma_{+}(q^{-\rho})P_{a-1}^{L_0}\Gamma_{+}(q^{-\rho})
      Q^{L_0}e^{H(\bst)}\\
  \mbox{}\times \Gamma_{-}(q^{-\rho})R_{b-1}^{L_0}\Gamma_{-}(q^{-\rho})\cdots 
      R_2^{L_0}\Gamma_{-}(q^{-\rho})P_1^{L_0}\Gamma_{-}(q^{-\rho})|s\rangle 
\label{Zab(s,t)-fermion}
\end{multline}
and 
\begin{multline}
  Z'_{a,b}(s,\bst,\bstbar) 
    = \langle s|\Gamma_{+}(q^{-\rho})P_1^{L_0}\Gamma_{+}(q^{-\rho})P_2^{L_0}\cdots 
      \Gamma_{+}(q^{-\rho})P_{a-1}^{L_0}\Gamma_{+}(q^{-\rho})
      Q^{L_0}e^{H(\bst,\bstbar)} \\
  \mbox{}\times \Gamma'_{-}(q^{-\rho})R_{b-1}^{L_0}\Gamma'_{-}(q^{-\rho})\cdots 
        R_2^{L_0}\Gamma'_{-}(q^{-\rho})P_1^{L_0}\Gamma'_{-}(q^{-\rho})|s\rangle, 
\label{Z'ab(s,t,tbar)-fermion}
\end{multline}
where 
\beqnn
  H(\bst) = \sum_{k=1}^\infty t_kH_k,\quad 
  \Hbar(\bstbar) = \sum_{k=1}^\infty\tbar_kH_{-k},\quad 
  H(\bst,\bstbar) = H(\bst) + \Hbar(\bstbar). 
\eeqnn

\section{Partition function as tau function: First orbifold model}

\subsection{Shift symmetries of quantum torus algebra}

Let $V^{(k)}_n$, $k,m\in\ZZ$, denote the fermion bilinears 
\beqnn
  V^{(k)}_m 
  = q^{-km/2}\sum_{n\in\ZZ}q^{kn}{:}\psi_{m-n}\psi^*_{n}{:} 
\eeqnn
that includes $H_k$'s and $J_k$'s as particular cases: 
\beqnn
  H_k = V^{(k)}_0,\quad J_k = V^{(0)}_k. 
\eeqnn
These fermion bilinears give a realization of 
(a central extension of) the 2D quantum torus algebra. 
Namely, they satisfy the commutation relations 
\begin{multline}
  [V^{(k)}_m, V^{(l)}_n]\\
  = \begin{cases}
    (q^{(lm-kn)/2} - q^{(kn-lm)/2})
    (V^{(k+l)}_{m+n} - \delta_{m+n,0}\frac{q^{k+l}}{1-q^{k+l}}) 
    &\text{if $k+l \not= 0$},\\
    (q^{-k(m+n)}-q^{k(m+n)})V^{(0)}_{m+n} + m\delta_{m+n,0} 
    &\text{if $k+l = 0$}.
    \end{cases}
\end{multline}
The case where $k = l = 0$ reduces to the commutation relations 
\beq
  [J_m,J_n] = m\delta_{m+n,0}
\eeq
of the $U(1)$ current algebra, whereas $H_k$'s are commutative.  

The shift symmetries \cite{NT07,NT08,Takasaki12,Takasaki13} 
act on these fermion bilinears.   Actually, 
there are three different types of shift symmetries: 
\begin{itemize}
\item[\rm (i)]For $k > 0$ and $m \in \ZZ$, 
\begin{multline}
  \Gamma_{+}(q^{-\rho})\left(
    V^{(k)}_m - \frac{q^k}{1-q^k}\delta_{m,0}
    \right)\Gamma_{+}(q^{-\rho})^{-1}\\
  = (-1)^k\Gamma_{-}(q^{-\rho})^{-1}\left(
        V^{(k)}_{m+k} - \frac{q^k}{1-q^k}\delta_{m+k,0}
        \right)\Gamma_{-}(q^{-\rho}). 
\label{SSi}
\end{multline}
\item[\rm (ii)] For $k > 0$ and $m \in \ZZ$, 
\begin{multline}
  \Gamma'_{+}(q^{-\rho})\left(
    V^{(-k)}_m + \frac{1}{1-q^k}\delta_{m,0}
    \right)\Gamma'_{+}(q^{-\rho})^{-1} \\
  = \Gamma'_{-}(q^{-\rho})^{-1}\left(
        V^{(-k)}_{m+k} + \frac{1}{1-q^k}\delta_{m+k,0}
        \right)\Gamma'_{-}(q^{-\rho}). 
\label{SSii}
\end{multline}
\item[\rm (iii)] For $k,m \in \ZZ$, 
\beq
  q^{W_0/2}V^{(k)}_mq^{-W_0/2} = V^{(k-m)}_m. 
\label{SSiii}
\eeq
\end{itemize}

These algebraic relations are used to convert 
the fermionic expression (\ref{Zab(s,t)-fermion}), 
(\ref{Z'ab(s,t,tbar)-fermion}) of the partition functions 
to tau functions.  This comprises two parts: 
\begin{itemize}
\item[1.] In the first part, $e^{H(\bst)}$ and $e^{\Hbar(\bstbar)}$ 
are moved to the leftmost or rightmost position 
of the operator product.  Note that $e^{H(\bst)}$ 
and $e^{\Hbar(\bstbar)}$ commute with $Q^{L_0}$. 
The first two sets (\ref{SSi}) and (\ref{SSii}) 
of the shift symmetries are used in this procedure. 
\item[2.] When the first part is finished, $H_k = V^{(k)}_0$ 
and $H_{-k} = V^{(-k)}_0$ in $e^{H(\bst)}$ and $e^{\Hbar(\bstbar)}$ 
turn into $V^{(k)}_{ak}$ and $V^{(-k)}_{-bk}$.  
In the second part, $V^{(k)}_{ak}$ and $V^{(-k)}_{-bk}$ 
are transformed to $V^{(0)}_{ak} = J_{ak}$ and $V^{(0)}_{-bk} = J_{-bk}$.
The operator $q^{W_0/2}$ in (\ref{SSiii}) is an avatar 
of ``framing factors'' in the topological vertex construction 
\cite{AKMV03}.  In the orbifold case, those framing factors 
are replaced by ``fractional'' ones \cite{BCY10}. 
We shall use such a variant of (\ref{SSiii}) 
in the subsequent calculations.  
\end{itemize}
Although mostly parallel to the case of the previous models 
\cite{NT07,NT08,Takasaki12,Takasaki13}, this procedure 
in the present setting becomes considerably complicated.  
This section is focused on the partition function 
$Z_{a,b}(s,\bst)$ of the first orbifold model. 
The partition function $Z'_{a,b}(s,\bst,\bstbar)$ 
of the second orbifold model is considered 
in the next section.

\subsection{Moving $e^{H(\bst)}$ towards the left}

Let us explain how to move 
\beqnn
  e^{H(\bst)} = \exp\left(\sum_{k=1}^\infty t_kV^{(0)}_k\right)
\eeqnn
towards the left through the product of 
$P_i^{L_0}\Gamma_{+}(q^{-\rho})$'s in (\ref{Zab(s,t)-fermion}).  
This is a repetition of processes in which 
the exponential operator overtakes 
each unit $P_i^{L_0}\Gamma_{+}(q^{-\rho})$ 
of the operator product.  In these processes, 
the exponent itself as well as the product 
of $P_i^{L_0}\Gamma_{+}(q^{-\rho})$'s are altered.  

At the first stage where 
$\exp\left(\sum_{k=1}^\infty t_kV^{(k)}_0\right)$ 
overtakes $P_{a-1}^{L_0}\Gamma_{+}(q^{-\rho})$, 
$\Gamma_{+}(q^{-\rho})$ turns into 
$\Gamma_{-}(q^{-\rho})\Gamma_{+}(q^{-\rho})$, 
and several $c$-number factors and an operator 
of the form $\Gamma_{-}(\cdots)^{-1}$ are generated. 

\begin{prop}
\begin{multline}
  \Gamma_{+}(q^{-\rho})P_1^{L_0}\cdots
    \Gamma_{+}(q^{-\rho})P_{a-1}^{L_0}\Gamma_{+}(q^{-\rho})
    \exp\left(\sum_{k=1}^\infty t_kV^{(k)}_0\right) \\
  = \exp\left(\sum_{k=1}^\infty \frac{t_kq^k}{1-q^k}\right) 
    \prod_{i=1}^{a-1} M(P_i\cdots P_{a-1},q)^{-1}
    \cdot\Gamma_{-}(P_1\cdots P_{a-1}q^{-\rho})^{-1}\\
  \mbox{}\times
    \Gamma_{+}(q^{-\rho})P_1^{L_0}\cdots
    \Gamma_{+}(q^{-\rho})P_{a-2}^{L_0}\Gamma_{+}(q^{-\rho})
    \exp\left(\sum_{k=1}^\infty t'_kV^{(k)}_k\right)\\
  \mbox{}\times 
    P_{a-1}^{L_0}\Gamma_{-}(q^{-\rho})\Gamma_{+}(q^{-\rho}),  
\label{e^H-to-left1}
\end{multline}
where 
\beqnn
  t'_k = (-1)^kP_{a-1}^{-k}t_k. 
\eeqnn
\end{prop}

\begin{proof}
The proof consists of three steps.  
First, use the first shift symmetry (\ref{SSi}) 
in the specialized form 
\beqnn
  \Gamma_{+}(q^{-\rho})V^{(k)}_0\Gamma_{+}(q^{-\rho})^{-1} 
  - \frac{q^k}{1-q^k} 
  = (-1)^k\Gamma_{-}(q^{-\rho})^{-1}V^{(k)}_k\Gamma_{-}(q^{-\rho}). 
\eeqnn
This implies the operator identity 
\begin{multline*}
  \Gamma_{+}(q^{-\rho})\exp\left(\sum_{k=1}^\infty t_kV^{(k)}_0\right) 
  = \exp\left(\sum_{k=1}^\infty \frac{t_kq^k}{1-q^k}\right)\\
  \mbox{}\times 
    \Gamma_{-}(q^{-\rho})^{-1}
    \exp\left(\sum_{k=1}^\infty(-1)^kt_kV^{(k)}_k\right)
    \Gamma_{-}(q^{-\rho})\Gamma_{+}(q^{-\rho}), 
\end{multline*}
hence 
\begin{multline*}
  \Gamma_{+}(q^{-\rho})P_1^{L_0}\cdots 
  \Gamma_{+}(q^{-\rho})P_{a-1}^{L_0}\Gamma_{+}(q^{-\rho})
  \exp\left(\sum_{k=1}^\infty t_kV^{(k)}_0\right) \\
  = \exp\left(\sum_{k=1}^\infty \frac{t_kq^k}{1-q^k}\right) 
    \Gamma_{+}(q^{-\rho})P_1^{L_0}\cdots
    \Gamma_{+}(q^{-\rho})P_{a-1}^{L_0}\\
  \mbox{}\times
    \Gamma_{-}(q^{-\rho})^{-1}
    \exp\left(\sum_{k=1}^\infty(-1)^kt_kV^{(k)}_k\right)
    \Gamma_{-}(q^{-\rho})\Gamma_{+}(q^{-\rho}). 
\end{multline*}
Second, move the newly generated operator 
$\Gamma_{-}(q^{-\rho})^{-1}$ towards the left.  
This can be achieved by repeated use 
of the scaling property (\ref{Gamma-scaling}) 
and the commutation relation 
\beq
  \Gamma_{+}(uq^{-\rho})\Gamma_{-}(vq^{-\rho})^{-1} 
  = M(uv,q)^{-1}\Gamma_{-}(vq^{-\rho})^{-1}\Gamma_{+}(uq^{-\rho}) 
\label{Gamma-CR}
\eeq
as follows: 
\begin{multline*}
  \Gamma_{+}(q^{-\rho})P_1^{L_0}\cdots
    \Gamma_{+}(q^{-\rho})P_{a-1}^{L_0}\Gamma_{-}(q^{-\rho})^{-1}\\
  = \prod_{i=1}^{a-1} M(P_i\cdots P_{a-1},q)^{-1}
    \cdot\Gamma_{-}(P_1\cdots P_{a-1}q^{-\rho})^{-1}\\
  \mbox{}\times
    \Gamma_{+}(q^{-\rho})P_1^{L_0}\cdots
    \Gamma_{+}(q^{-\rho})P_{a-2}^{L_0}\Gamma_{+}(q^{-\rho})P_{a-1}^{L_0}. 
\end{multline*}
Lastly, use the scaling property 
\beq
  u^{L_0}V^{(k)}_mu^{-L_0} = u^{-m}V^{(k)}_m
\label{V-scaling}
\eeq
of $V^{(k)}_m$'s to move $P_{a-1}^{L_0}$ to the right 
of the exponential operator as 
\beqnn
  P_{a-1}^{L_0}\exp\left(\sum_{k=1}^\infty(-1)^kt_kV^{(k)}_k\right) 
  = \exp\left(\sum_{k=1}^\infty(-1)^kP_{a-1}^{-k}t_kV^{(k)}_k\right) 
    P_{a-1}^{L_0}. 
\eeqnn
\end{proof}

At the next stage, 
$\exp\left(\sum_{k=1}^\infty t'_kV^{(k)}_k\right)$ 
overtakes $P_{a-2}^{L_0}\Gamma_{+}(q^{-\rho})$. 
the calculations are fully parallel to the first stage. 
The first step is to use the operator identity 
\beqnn
  \Gamma_{+}(q^{-\rho})\exp\left(\sum_{k=1}^\infty t'_kV^{(k)}_k\right) 
  = \Gamma_{-}(q^{-\rho})^{-1}
    \exp\left(\sum_{k=1}^\infty (-1)^kt'_kV^{(k)}_{2k}\right) 
    \Gamma_{-}(q^{-\rho})\Gamma_{+}(q^{-\rho}) 
\eeqnn
that follows from the specialized shift symmetry 
\beqnn
  \Gamma_{+}(q^{-\rho})V^{(k)}_k\Gamma_{+}(q^{-\rho})^{-1} 
  = (-1)^k\Gamma_{-}(q^{-\rho})^{-1}V^{(k)}_{2k}\Gamma_{-}(q^{-\rho}). 
\eeqnn
The second and third steps are to move $\Gamma_{-}(q^{-\rho})^{-1}$ 
towards the left end and $P_{a-2}^{L_0}$ to the right 
of the exponential operator.  This leads to the algebraic relation 
\begin{multline*}
  \Gamma_{+}(q^{-\rho})P_1^{L_0}\cdots
  \Gamma_{+}(q^{-\rho})P_{a-2}^{L_0}\Gamma_{+}(q^{-\rho})
  \exp\left(\sum_{k=1}^\infty t'_kV^{(k)}_k\right) \\
  = \prod_{i=1}^{a-2} M(P_i\cdots P_{a-2},q)^{-1} 
    \cdot\Gamma_{-}(P_1\cdots P_{a-2}q^{-\rho})^{-1}\\
  \mbox{}\times
    \Gamma_{+}(q^{-\rho})P_1^{L_0}\cdots
    \Gamma_{+}(q^{-\rho})P_{a-3}^{L_0}\Gamma_{+}(q^{-\rho})
    \exp\left(\sum_{k=1}^\infty t''_kV^{(k)}_{2k}\right)\\
  \mbox{}\times 
    P_{a-2}^{L_0}\Gamma_{-}(q^{-\rho})\Gamma_{+}(q^{-\rho}), 
\end{multline*}
where
\beqnn
  t''_k = (-1)^kP_{a-2}^{-2k}t'_k. 
\eeqnn

The overtaking process can be repeated 
until the transformed exponential operator 
overtakes all $\Gamma_{+}(q^{-\rho})$'s and $P_i^{L_0}$'s. 
The net result reads as follows. 

\begin{prop}
\begin{multline}
  \Gamma_{+}(q^{-\rho})P_1^{L_0}\cdots
  \Gamma_{+}(q^{-\rho})P_{a-1}^{L_0}\Gamma_{+}(q^{-\rho})
  \exp\left(\sum_{k=1}^\infty t_kV^{(k)}_0\right) \\
  = \exp\left(\sum_{k=1}^\infty\frac{t_kq^k}{1-q^k}\right)
    \prod_{1\le i\le j\le a-1}M(P_i\cdots P_j,q)^{-1}\\
  \mbox{}\times 
    \Gamma_{-}(q^{-\rho})^{-1}
    \prod_{j=1}^{a-1}\Gamma_{-}(P_1\cdots P_jq^{-\rho})^{-1}\cdot
    \exp\left(\sum_{k=1}^\infty T_kV^{(k)}_{ak}\right)\\
  \mbox{}\times 
    \Gamma_{-}(q^{-\rho})\Gamma_{+}(q^{-\rho})P_1^{L_0}\cdots
    \Gamma_{-}(q^{-\rho})\Gamma_{+}(q^{-\rho})P_{a-1}^{L_0}
    \Gamma_{-}(q^{-\rho})\Gamma_{+}(q^{-\rho}), 
\label{e^H-to-left2}
\end{multline}
where 
\beq
  T_k = (-1)^{ak}P_1^{-(a-1)k}P_2^{-(a-2)k}\cdots P_{a-1}^{-k}t_k. 
\label{T_k-t_k}
\eeq
\end{prop}

Thus, as $e^{H(\bst)}$ in (\ref{Zab(s,t)-fermion}) 
is transferred to the left end of the operator product, 
the terms $t_kH_k$ in $H(\bst)$ are transformed to $T_kV^{(k)}_{ak}$.  
One can thereby rewrite the left half of the fermionic expression 
(\ref{Zab(s,t)-fermion}) as 
\begin{multline}
  \langle s|\Gamma_{+}(q^{-\rho})P_1^{L_0}\cdots
  \Gamma_{+}(q^{-\rho})P_{a-1}^{L_0}\Gamma_{+}(q^{-\rho})
  Q^{L_0}e^{H(\bst)}\\
  = \exp\left(\sum_{k=1}^\infty\frac{t_kq^k}{1-q^k}\right)
    \prod_{1\le i\le j\le a-1}M(P_i\cdots P_j,q)^{-1}\cdot
    \langle s|\exp\left(\sum_{k=1}^\infty T_kV^{(k)}_{ak}\right)\\
  \mbox{}\times 
    \Gamma_{-}(q^{-\rho})\Gamma_{+}(q^{-\rho})P_1^{L_0}\cdots 
    \Gamma_{-}(q^{-\rho})\Gamma_{+}(q^{-\rho})P_{a-1}^{L_0}
    \Gamma_{-}(q^{-\rho})\Gamma_{+}(q^{-\rho})Q^{L_0}. 
\label{e^H-to-left3}
\end{multline}
Note that $\Gamma_{-}^{-1}$'s disappear because 
$\langle s|J_{-k} = 0$ for $k > 0$.

\subsection{Converting $Z_{a,b}(s,\bst)$ to a tau function}

To interpret $Z_{a,b}(s,\bst)$ as a tau function, 
$V^{(k)}_{ak}$'s in (\ref{e^H-to-left3}) have to be 
transformed to $J_{ak}$'s.  This can be achieved 
by the following variant of the third shift symmetry (\ref{SSiii}). 

\begin{lemma}
For $k \in \ZZ$, 
\beq
  q^{W_0/2a}V^{(k)}_{ak}q^{-W_0/2a} = J_{ak}. 
  \label{modSSiii}
\eeq
\end{lemma}

\begin{proof}
\begin{align*}
  q^{W_0/2a}V^{(k)}_{ak}q^{-W_0/2a} 
  &= q^{-ak^2/2}\sum_{n\in\ZZ}q^{kn}
       q^{W_0/2a}{:}\psi_{ak-n}\psi^*_n{:}q^{-W_0/2a}\\
  &= q^{-ak^2/2}\sum_{n\in\ZZ}
       q^{kn+(ak-n)^2/2a - n^2/2a}{:}\psi_{ak-n}\psi^*_n{:}\\
  &= J_{ak}. 
\end{align*}
\end{proof}

By (\ref{modSSiii}), one can substitute 
\beqnn
  V^{(k)}_{ak} = q^{-W_0/2a}J_{ak}q^{W_0/2a}
\eeqnn
in (\ref{e^H-to-left3}).  Since $q^{-W_0/2a}$ turns into 
a c-number factor as 
\beqnn
  \langle s|q^{-W_0/2a} = q^{-s(s+1)(2s+1)/12a}\langle s|, 
\eeqnn
one can rewrite (\ref{e^H-to-left3}) as 
\begin{multline}
  \langle s|\Gamma_{+}(q^{-\rho})P_1^{L_0}\cdots 
  \Gamma_{+}(q^{-\rho})P_{a-1}^{L_0}\Gamma_{+}(q^{-\rho})
  Q^{L_0}e^{H(\bst)}\\
  = \exp\left(\sum_{k=1}^\infty\frac{t_kq^k}{1-q^k}\right)q^{-s(s+1)(2s+1)/12a}
    \prod_{1\le i\le j\le a-1}M(P_i\cdots P_j,q)^{-1}\\
  \mbox{}\times 
    \langle s|\exp\left(\sum_{k=1}^\infty T_kJ_{ak}\right)
    q^{W_0/2a}\Gamma_{-}(q^{-\rho})\Gamma_{+}(q^{-\rho})P_1^{L_0}\cdots\\
  \mbox{}\times
    \Gamma_{-}(q^{-\rho})\Gamma_{+}(q^{-\rho})P_{a-1}^{L_0}
    \Gamma_{-}(q^{-\rho})\Gamma_{+}(q^{-\rho})Q^{L_0}. 
\label{e^H-to-left4} 
\end{multline}

This expression of the left half of (\ref{Zab(s,t)-fermion}) 
shows that $Z_{a,b}(s,\bst)$ is essentially a tau function 
of the KP hierarchy.  To see a link with the bi-graded Toda hierarchy, 
however, the right half of (\ref{Zab(s,t)-fermion}), too, 
has to be rewritten as follows.  

\begin{lemma}
\begin{multline}
  \Gamma_{-}(q^{-\rho})R_{b-1}^{L_0}\Gamma_{-}(q^{-\rho})\cdots 
  R_1^{L_0}\Gamma_{-}(q^{-\rho})|s\rangle\\
  = \Gamma_{-}(q^{-\rho})\Gamma_{+}(q^{-\rho})
    R_{b-1}^{L_0}\Gamma_{-}(q^{-\rho})\Gamma_{+}(q^{-\rho})\cdots 
    R_1^{L_0}\Gamma_{-}(q^{-\rho})\Gamma_{+}(q^{-\rho})q^{W_0/2b}|s\rangle\\
  \mbox{}\times 
    q^{-s(s+1)(2s+1)/12b}\prod_{1\le i\le j\le b-1}M(R_i\cdots R_j,q)^{-1}. 
\label{e^H-to-left5}
\end{multline}
\end{lemma}

\begin{proof}
Insert $\Gamma_{+}(q^{-\rho})\Gamma_{+}(q^{-\rho})^{-1}$ 
to the right of each of $\Gamma_{-}(q^{-\rho})$'s as 
\begin{multline*}
  \Gamma_{-}(q^{-\rho})R_{b-1}^{L_0}\Gamma_{-}(q^{-\rho})
  \cdots R_1^{L_0}\Gamma_{-}(q^{-\rho})|s\rangle\\
  = \Gamma_{-}(q^{-\rho})\Gamma_{+}(q^{-\rho})\Gamma_{+}(q^{-\rho})^{-1}R_{b-1}^{L_0}
    \Gamma_{-}(q^{-\rho})\Gamma_{+}(q^{-\rho})\Gamma_{+}(q^{-\rho})^{-1}\cdots\\
  \mbox{}\times 
    R_1^{L_0}\Gamma_{-}(q^{-\rho})\Gamma_{+}(q^{-\rho})\Gamma_{+}(q^{-\rho})^{-1}
   |s\rangle. 
\end{multline*}
The rightmost $\Gamma_{+}(q^{-\rho})^{-1}$ disappears 
as it hits $|s\rangle$.  By the commutation relation (\ref{Gamma-CR}), 
one can move the other $\Gamma_{+}(q^{-\rho})^{-1}$'s 
towards the right end of the operator product. 
The transformed operators $\Gamma_{+}(R_1\cdots R_jq^{-\rho})^{-1}$, 
$j = 1,\ldots,b-1$, hit $|s\rangle$ and disappear.  
The outcome reads 
\begin{multline*}
  \Gamma_{-}(q^{-\rho})R_{b-1}^{L_0}\Gamma_{-}(q^{-\rho})\cdots 
  R_1^{L_0}\Gamma_{-}(q^{-\rho})|s\rangle\\
  = \Gamma_{-}(q^{-\rho})\Gamma_{+}(q^{-\rho})
    R_{b-1}^{L_0}\Gamma_{-}(q^{-\rho})\Gamma_{+}(q^{-\rho})\cdots
    R_1^{L_0}\Gamma_{-}(q^{-\rho})\Gamma_{+}(q^{-\rho})|s\rangle\\
  \mbox{}\times
    \prod_{1\le i\le j\le b-1}M(R_i\cdots R_j,q)^{-1}. 
\end{multline*}
Lastly, rewriting $|s\rangle$ as 
\beqnn
  |s\rangle = q^{W_0/2b}|s\rangle q^{-s(s+1)(2s+1)/12b}, 
\eeqnn
one obtains (\ref{e^H-to-left5}). 
\end{proof}

The inner product of (\ref{e^H-to-left4}) and (\ref{e^H-to-left5}) 
yields the following final expression of $Z_{a,b}(s,\bst)$. 
\begin{prop}
\beq
  Z_{a,b}(s,\bst) = f_{a,b}(s,\bst)
    \langle s|\exp\left(\sum_{k=1}^\infty T_kJ_{ak}\right)g|s\rangle, 
\label{Zab(s,t)-tau1}
\eeq
where 
\begin{multline}
  f_{a,b}(s,\bst) 
  = \exp\left(\sum_{k=1}^\infty\frac{t_kq^k}{1-q^k}\right)
    q^{-s(s+1)(2s+1)(1/12a+1/12b)}\\
  \mbox{}\times 
    \prod_{1\le i\le j\le a-1}M(P_i\cdots P_j,q)^{-1}\cdot 
    \prod_{1\le i\le j\le b-1}M(R_i\cdots R_j,q)^{-1} 
\label{fab(s,t)}
\end{multline}
and 
\begin{multline}
  g = q^{W_0/2a}\Gamma_{-}(q^{-\rho})\Gamma_{+}(q^{-\rho})P_1^{L_0}\cdots 
    \Gamma_{-}(q^{-\rho})\Gamma_{+}(q^{-\rho})
    P_{a-1}^{L_0}\Gamma_{-}(q^{-\rho})\Gamma_{+}(q^{-\rho})Q^{L_0}\\
  \mbox{}\times 
    \Gamma_{-}(q^{-\rho})\Gamma_{+}(q^{-\rho})
    R_{b-1}^{L_0}\Gamma_{-}(q^{-\rho})\Gamma_{+}(q^{-\rho})\cdots 
    R_1^{L_0}\Gamma_{-}(q^{-\rho})\Gamma_{+}(q^{-\rho})q^{W_0/2b}. 
\label{g}
\end{multline}
\end{prop}

Actually, one can derive another expression of $Z_{ab}(s,\bst)$ 
based on the same operator $g$ as follows.  This is a reason 
why we modify the right half of (\ref{Zab(s,t)-fermion}) 
as shown in (\ref{e^H-to-left5}).  

\begin{prop}
\beq
  Z_{a,b}(s,\bst) = f_{a,b}(s,\bst) 
    \langle s|g\exp\left(\sum_{k=1}^\infty \Tbar_kJ_{-bk}\right)|s\rangle, 
\label{Zab(s,t)-tau2}
\eeq
where 
\beq
  \Tbar_k = (-1)^{bk}R_1^{-(b-1)k}R_2^{-(b-2)k}\cdots R_{b-1}^{-k}t_k. 
\label{Tbar_k-t_k}
\eeq
\end{prop}

\begin{proof}
Derivation of (\ref{Zab(s,t)-tau2}) is fully parallel 
to the case of (\ref{Zab(s,t)-tau1}) except that 
$e^{H(\bst)}$ is moved {\it towards the right}.  
Start from the specialization 
\beqnn
  \Gamma_{+}(q^{-\rho})V^{(k)}_{-k}\Gamma_{+}(q^{-\rho})^{-1}
  = (-1)^k\left(\Gamma_{-}(q^{-\rho})^{-1}V^{(k)}_0\Gamma_{-}(q^{-\rho})
           - \frac{q^k}{1-q^k}\right) 
\eeqnn
of the first shift symmetry (\ref{SSi}).  This implies that 
\begin{multline*}
  \exp\left(\sum_{k=1}^\infty t_kV^{(k)}_0\right)\Gamma_{-}(q^{-\rho}) 
  = \left(\sum_{k=1}^\infty\frac{t_kq^k}{1-q^k}\right)\\
  \mbox{}\times
    \Gamma_{-}(q^{-\rho})\Gamma_{+}(q^{-\rho})
    \exp\left(\sum_{k=1}^\infty(-1)^kt_kV^{(k)}_{-k}\right)
    \Gamma_{+}(q^{-\rho})^{-1}. 
\end{multline*}
Sending $\Gamma_{+}(q^{-\rho})^{-1}$ towards the right 
beyond the product of $R_i^{L_0}$'s and $\Gamma_{-}(q^{-\rho})$'s, 
one obtains the algebraic relation 
\begin{multline*}
    \exp\left(\sum_{k=1}^\infty t_kV^{(k)}_0\right) 
    \Gamma_{-}(q^{-\rho})R_{b-1}^{L_0}\Gamma_{-}(q^{-\rho})
    \cdots R_1^{L_0}\Gamma_{-}(q^{-\rho})\\
  = \exp\left(\sum_{k=1}^\infty \frac{t_kq^k}{1-q^k}\right) 
    \prod_{i=1}^{a-1} M(R_i\cdots R_{b-1},q)^{-1}\cdot 
    \Gamma_{-}(q^{-\rho})\Gamma_{+}(q^{-\rho})R_{b-1}^{L_0}\\
  \mbox{}\times
    \exp\left(\sum_{k=1}^\infty\tbar'_kV^{(k)}_{-k}\right)
    \Gamma_{-}(q^{-\rho})R_{b-2}^{L_0}\Gamma_{-}(q^{-\rho})
    \cdots R_1^{L_0}\Gamma_{-}(q^{-\rho})\\
  \mbox{}\times
    \Gamma_{+}(R_1\cdots R_{b-1}q^{-\rho})^{-1}, 
\end{multline*}
where 
\beqnn
  t'_k = (-1)^kR_{b-1}^{-k}t_k. 
\eeqnn
This relation amounts to (\ref{e^H-to-left1}).  
Repeating this transfer procedure leads 
to the following counterpart of (\ref{e^H-to-left2}): 
\begin{multline*}
  \exp\left(\sum_{k=1}^\infty t_kV^{(k)}_0\right) 
  \Gamma_{-}(q^{-\rho})R_{b-1}^{L_0}\Gamma_{-}(q^{-\rho})
  \cdots R_1^{L_0}\Gamma_{-}(q^{-\rho})\\
  = \exp\left(\sum_{k=1}^\infty\frac{t_kq^k}{1-q^k}\right)
    \prod_{1\le i\le j\le b-1}M(R_i\cdots R_j,q)^{-1}\\
  \mbox{}\times 
    \Gamma_{-}(q^{-\rho})\Gamma_{+}(q^{-\rho})
    R_{b-1}^{L_0}\Gamma_{-}(q^{-\rho})\Gamma_{+}(q^{-\rho})\cdots
    R_1^{L_0}\Gamma_{-}(q^{-\rho})\Gamma_{+}(q^{-\rho})\\
  \mbox{}\times 
    \exp\left(\sum_{k=1}^\infty \Tbar_kV^{(k)}_{-bk}\right)
    \Gamma_{+}(q^{-\rho})^{-1}
    \prod_{j=1}^{b-1}\Gamma_{+}(R_1\cdots R_jq^{-\rho})^{-1}. 
\end{multline*}
In much the same way, one can derive counterparts 
of (\ref{e^H-to-left3}), (\ref{e^H-to-left4}), 
(\ref{e^H-to-left5}).   Equation (\ref{Zab(s,t)-tau2}) 
is an immediate consequence of these algebraic relations. 
\end{proof}

Equations (\ref{Zab(s,t)-tau1}) and (\ref{Zab(s,t)-tau2}) 
show that the partition function $Z_{a,b}(s,\bst)$ 
of the first orbifold model is related to the tau function 
\beq
  \tau(s,\bst,\bstbar) 
  = \langle s|\exp\left(\sum_{k=1}^\infty t_kJ_k\right)g
    \exp\left(- \sum_{k=1}^\infty\tbar_kJ_{-k}\right)|s\rangle
\label{tau(g)}
\eeq
of the 2D Toda hierarchy \cite{Takebe91} 
restricted to a subspace of the full time evolutions.  
The existence of the two expressions (\ref{Zab(s,t)-tau1}) 
and (\ref{Zab(s,t)-tau2}) implies the identity 
\beq
    \langle s|\exp\left(\sum_{k=1}^\infty T_kJ_{ak}\right)g|s\rangle 
  = \langle s|g\exp\left(\sum_{k=1}^\infty \Tbar_kJ_{-bk}\right)|s\rangle 
\label{<e^Tg>=<ge^T>}
\eeq
of the restricted tau functions or, equivalently, 
the operator identities 
\begin{multline}
  (-1)^{ak}P_1^{-(a-1)k}P_2^{-(a-2)k}\cdots P_{a-1}^{-k}J_{ak}g \\
  = (-1)^{bk}R_1^{-(b-1)k}R_2^{-(b-2)k}\cdots R_{b-1}^{-k} gJ_{-bk}, 
  \quad k = 1,2,\ldots. 
\label{Jg=gJ}
\end{multline}
These algebraic relations characterize tau functions 
of the bi-graded Toda hierarchy of bi-degree $(a,b)$ 
\cite{Kuperschmidt85,Carlet06}.  
We shall reconfirm this fact in the Lax formalism. 

In summary, we have proved the following relation 
to the bi-graded Toda hierarchy. 

\begin{theorem}
The partition function $Z_{a,b}(s,\bst)$ is related 
to the tau function (\ref{tau(g)}) of the 2D Toda hierarchy 
in two ways as 
\beq
  Z_{a,b}(s,\bst) 
    = f_{a,b}(s,\bst)\tau(s,\bsT,\bszero) 
    = f_{a,b}(s,\bst)\tau(s,\bszero,-\bsTbar), 
\label{Zab(s,t)-tau3}
\eeq
where 
\beqnn
\begin{gathered}
  \bsT = (\underbrace{0,\ldots,0}_{a-1},T_1,
          \underbrace{0,\ldots,0}_{a-1},T_2,\ldots
          \underbrace{0,\ldots,0}_{a-1},T_k,\ldots),\\
  \bsTbar = (\underbrace{0,\cdots,0}_{b-1},\Tbar_1,
             \underbrace{0,\cdots,0}_{b-1},\Tbar_2,\ldots,
             \underbrace{0,\cdots,0}_{b-1},\Tbar_k,\ldots).  
\end{gathered}
\eeqnn
$T_k$'s and $\Tbar_k$'s are obtained from $t_k$'s as shown 
in (\ref{T_k-t_k}) and (\ref{Tbar_k-t_k}).   
The prefactor $f_{a,b}(s,\bst)$ is built from exponential 
and MacMahon functions as shown in (\ref{fab(s,t)}).  
The generating operator (\ref{g}) of the tau function 
satisfies the algebraic relations (\ref{Jg=gJ}) 
that characterize the bi-graded Toda hierarchy of bi-degree $(a,b)$.  
\end{theorem}

\begin{remark}
In the non-orbifold ($a = b = 1$) case, 
(\ref{<e^Tg>=<ge^T>}) reduces to the identity 
\beqnn
    \langle s|\exp\left(\sum_{k=1}^\infty t_kJ_{ak}\right)g|s\rangle 
  = \langle s|g\exp\left(\sum_{k=1}^\infty t_kJ_{-bk}\right)|s\rangle, 
\eeqnn
which implies that $\tau(s,\bst,\bstbar)$ is a function 
of $\bst-\bstbar$, 
\beqnn
  \tau(s,\bst,\bstbar) = \tau(s,\bst-\bstbar). 
\eeqnn
The reduced function $\tau(s,\bst)$ is a tau function 
of the 1D Toda hierarchy.  Moreover, (\ref{Zab(s,t)-tau3}) 
takes the simpler form
\beqnn
  Z_{1,1}(s,\bst) = f_{1,1}(s,\bst)\tau(s,\iota(\bst),\bszero) 
  = f_{1,1}(s,\bst)\tau(s,\bszero,-\iota(\bst)), 
\eeqnn
where 
\beqnn
  \iota(\bst) = (-t_1,t_2,-t_3,\ldots,(-1)^kt_k,\ldots),
\eeqnn
thus both $\bsT$ and $\bsTbar$ coincide with $\iota(\bst)$. 
\end{remark}

\begin{remark}
In the orbifold model, the reduced time variables 
$\bsT$ and $\bsTbar$ of the tau function are full of {\it gaps\/}. 
Namely, the coupling constants $t_k$, $k = 1,2,\ldots$ 
of the orbifold model show up only at every $a$-steps in $\bsT$ 
and every $b$-steps $\bsTbar$, and the other components 
of $\bsT$ and $\bsTbar$ are set to $0$.  This is related 
to the structure of the Lax operators in the Lax formalism. 
\end{remark}

\begin{remark}
The operators $q^{W_0/2a}$ and $q^{W_0/2b}$ in the definition 
(\ref{g}) of $g$ are avatars of fractional framing factors 
in the orbifold topological vertex construction \cite{BCY10}. 
\end{remark}

\begin{remark}
As mentioned in Remark \ref{remark-2q}, 
our previous melting crystal model 
with two $q$-parameters is a special case 
of the present orbifold model.  This explains 
why the same bi-graded Toda hierarchy emerges therein 
when the $q$-parameters are specialized as shown 
in (\ref{q1q2-special}) \cite{Takasaki09}. 
\end{remark}

\section{Partition function as tau function: Second orbifold model}

\subsection{Moving $e^{\Hbar(\bstbar)}$ towards the right}

To convert the fermionic expression (\ref{Z'ab(s,t,tbar)-fermion}) 
of $Z'_{a,b}(s,\bst,\bstbar)$ to a tau function, 
we split $e^{H(\bst,\bstbar)}$ as 
\beqnn
  e^{H(\bst,\bstbar)} = e^{H(\bst)}e^{\Hbar(\bstbar)}
\eeqnn
and move $e^{H(\bst)}$ and $e^{\Hbar(\bstbar)}$ 
towards the left and the right, respectively.  
Since the transfer procedure for $e^{H(\bst)}$ 
is exactly the same as in the case of $Z_{a,b}(s,\bst)$, 
let us explain how to move 
\beqnn
  e^{\Hbar(\bstbar)} 
  = \exp\left(\sum_{k=1}^\infty\tbar_kV^{(-k)}_0\right) 
\eeqnn
towards the right.  The whole procedure is parallel 
to the case of $e^{H(\bst)}$. 

The following algebraic relation amounts to (\ref{e^H-to-left1}). 

\begin{prop}
\begin{multline}
  \exp\left(\sum_{k=1}^\infty\tbar_kV^{(-k)}_0\right)
  \Gamma'_{-}(q^{-\rho})R_{b-1}^{L_0}\Gamma'_{-}(q^{-\rho})
  \cdots R_1^{L_0}\Gamma'_{-}(q^{-\rho}) \\
  = \exp\left(- \sum_{k=1}^\infty\frac{\tbar_k}{1-q^k}\right) 
    \prod_{i=1}^{a-1}M(R_1\cdots R_j,q)^{-1}\cdot 
    \Gamma'_{-}(q^{-\rho})\Gamma'_{+}(q^{-\rho})R_{b-1}^{L_0}\\
  \mbox{}\times
    \exp\left(\sum_{k=1}^\infty\tbar'_kV^{(-k)}_{-k}\right)
    \Gamma'_{-}(q^{-\rho})R_{b-2}^{L_0}\Gamma'_{-}(q^{-\rho})
    \cdots R_1^{L_0}\Gamma'_{-}(q^{-\rho})\\
  \mbox{}\times 
    \Gamma'_{+}(R_1\cdots R_{b-1}q^{-\rho})^{-1}, 
\label{e^Hbar-to-right1}
\end{multline}
where 
\beqnn
  \tbar'_k = R_{b-1}^{-k}\tbar_k. 
\eeqnn
\end{prop}

\begin{proof}
The proof is parallel to that of (\ref{e^H-to-left1}). 
First, use the second shift symmetry (\ref{SSii}) 
in the specialized form
\beqnn
  \Gamma'_{+}(q^{-\rho})V^{(-k)}_{-k}\Gamma'_{+}(q^{-\rho})^{-1} 
  = \Gamma'_{-}(q^{-\rho})^{-1}V^{(-k)}_0\Gamma'_{-}(q^{-\rho}) 
    + \frac{1}{1-q^k}. 
\eeqnn
This implies the operator identity 
\begin{multline*}
  \exp\left(\sum_{k=1}^\infty\tbar_kV^{(-k)}_0\right)\Gamma'_{-}(q^{-\rho})
  = \exp\left(-\sum_{k=1}^\infty\frac{\tbar_k}{1-q^k}\right)\\
  \mbox{}\times 
    \Gamma'_{-}(q^{-\rho})\Gamma'_{+}(q^{-\rho})
    \exp\left(\sum_{k=1}^\infty\tbar_kV^{(-k)}_{-k}\right)
    \Gamma'_{+}(q^{-\rho})^{-1}, 
\end{multline*}
hence, 
\begin{multline*}
  \exp\left(\sum_{k=1}^\infty\tbar_kV^{(-k)}_0\right)
  \Gamma'_{-}(q^{-\rho})R_{b-1}^{L_0}\Gamma'_{-}(q^{-\rho})
  \cdots R_1^{L_0}\Gamma'_{-}(q^{-\rho}) \\
  = \exp\left(- \sum_{k=1}^\infty\frac{\tbar_k}{1-q^k}\right) 
    \Gamma'_{-}(q^{-\rho})\Gamma'_{+}(q^{-\rho})
    \exp\left(\sum_{k=1}^\infty\tbar_kV^{(-k)}_{-k}\right)
    \Gamma'_{+}(q^{-\rho})^{-1}\\
  \mbox{}\times
    R_{b-1}^{L_0}\Gamma'_{-}(q^{-\rho})\cdots 
    R_1^{L_0}\Gamma'_{-}(q^{-\rho}). 
\end{multline*}
Second, use the commutation relation 
\beq
  \Gamma'_{-}(uq^{-\rho})^{-1}\Gamma'_{+}(vq^{-\rho})
  = M(uv,q)^{-1}\Gamma'_{+}(vq^{-\rho})\Gamma'_{-}(uq^{-\rho})^{-1}
\label{Gamma'-CR}
\eeq
of the primed operators to move the newly generated 
operator $\Gamma'_{+}(q^{-\rho})^{-1}$ towards the right as 
\begin{multline*}
  \Gamma'_{+}(q^{-\rho})^{-1}R_{b-1}^{L_0}\Gamma'_{-}(q^{-\rho})
  \cdots R_1\Gamma'_{-}(q^{-\rho})\\
  = \prod_{i=1}^{b-1}M(R_i\cdots R_{b-1},q)^{-1}\cdot 
    R_{b-1}^{L_0}\Gamma'_{-}(q^{-\rho})\cdots 
    R_1^{L_0}\Gamma'_{-}(q^{-\rho})\\
  \mbox{}\times
    \Gamma'_{+}(R_1\cdots R_{b-1}q^{-\rho})^{-1}. 
\end{multline*}
Lastly, use the scaling property (\ref{V-scaling}) 
to move $R_{b-1}^{L_0}$ to the left of 
the exponential operator as 
\beqnn
  \exp\left(\sum_{k=1}^\infty\tbar_kV^{(-k)}_{-k}\right)R_{b-1}^{L_0} 
  = R_{b-1}^{L_0}
    \exp\left(\sum_{k=1}^\infty R_{b-1}^{-k}\tbar_kV^{(-k)}_{-k}\right). 
\eeqnn
\end{proof}

Thus as $\exp\left(\sum_{k=1}^\infty\tbar_kV^{(-k)}_{-k}\right)$ 
overtakes $\Gamma'_{-}(q^{-\rho})R_{b-1}^{L_0}$, 
$\tbar_kV^{(-k)}$ and $\Gamma'_{-}(q^{-\rho})$ 
turns into $\tbar'_kV^{(-k)}_{-k}$ and 
$\Gamma'_{-}(q^{-\rho})\Gamma'_{+}(q^{-\rho})$. 
Moreover, several $c$-number factors and 
an operator factor of the form $\Gamma'_{+}(\cdots)^{-1}$ 
are generated.  

Repeating this overtaking process, one obtains 
the following counterpart of (\ref{e^H-to-left2}). 

\begin{prop}
\begin{multline}
  \exp\left(\sum_{k=1}^\infty\tbar_kV^{(-k)}_0\right)
  \Gamma'_{-}(q^{-\rho})R_{b-1}^{L_0}\Gamma'_{-}(q^{-\rho})
  \cdots R_1^{L_0}\Gamma'_{-}(q^{-\rho}) \\
  = \exp\left(- \sum_{k=1}^\infty\frac{\tbar_k}{1-q^k}\right) 
    \prod_{1\le i\le j\le b-1}M(R_i\cdots R_j,q)^{-1}\\
  \mbox{}\times 
    \Gamma'_{-}(q^{-\rho})\Gamma'_{+}(q^{-\rho})
    R_{b-1}^{L_0}\Gamma'_{-}(q^{-\rho})\Gamma'_{+}(q^{-\rho})\cdots
    R_1^{L_0}\Gamma'_{-}(q^{-\rho})\Gamma'_{+}(q^{-\rho})\\
  \mbox{}\times 
    \exp\left(\sum_{k=1}^\infty\Tbar_kV^{(-k)}_{-bk}\right)
    \Gamma'_{+}(q^{-\rho})^{-1}
    \prod_{j=1}^{b-1}\Gamma'_{+}(R_1\cdots R_je^{-\rho})^{-1}, 
\label{e^Hbar-to-right2}
\end{multline}
where 
\beq
  \Tbar_k = R_1^{-(b-1)k}R_2^{-(b-2)k}\cdots R_{b-1}^{-k}\tbar_k. 
  \label{Tbar_k-tbar_k}
\eeq
\end{prop}

\subsection{Converting $Z'_{a,b}(s,\bst,\bstbar)$ to a tau function}

We now proceed to converting (\ref{Z'ab(s,t,tbar)-fermion}) 
to a tau function. Since the left half of this expression 
is the same as that of (\ref{Zab(s,t)-fermion}), 
one can use (\ref{e^H-to-left4}) as it is.  

As regards the right half of (\ref{Z'ab(s,t,tbar)-fermion}), 
the forgoing operator identity (\ref{e^Hbar-to-right2})
implies that 
\begin{multline}
  e^{\Hbar(\bstbar)}\Gamma'_{-}(q^{-\rho})
  R_{b-1}^{L_0}\Gamma'_{-}(q^{-\rho})\cdots 
  R_1^{L_0}\Gamma'_{-}(q^{-\rho})|s\rangle\\
  = \exp\left(-\sum_{k=1}^\infty\frac{\tbar_k}{1-q^k}\right)
    \prod_{1\le i\le j\le b-1}M(R_i\cdots R_j,q)^{-1}\\
  \mbox{}\times
    \Gamma'_{-}(q^{-\rho})\Gamma'_{+}(q^{-\rho})
    R_{b-1}^{L_0}\Gamma'_{-}(q^{-\rho})\Gamma'_{+}(q^{-\rho})\cdots\\
  \mbox{}\times
    R_1^{L_0}\Gamma'_{-}(q^{-\rho})\Gamma'_{+}(q^{-\rho})
    \exp\left(\sum_{k=1}^\infty \Tbar_kV^{(-k)}_{-bk}\right)|s\rangle. 
\label{e^Hbar-to-right3}
\end{multline}
Moreover, one can use the identity 
\beq
  q^{W_0/2b}V^{(-k)}_{-bk}q^{-W_0/2b} = J_{-bk} 
\label{modSSiii'}
\eeq
in place of (\ref{modSSiii}) to rewrite the last part 
of (\ref{e^Hbar-to-right3}) as 
\begin{align*}
  \exp\left(\sum_{k=1}^\infty \Tbar_kV^{(-k)}_{-bk}\right)|s\rangle
  &= q^{-W_0/2b}\exp\left(\sum_{k=1}^\infty\Tbar_kJ_{-bk})\right)
     q^{W_0/2b}|s\rangle\\
  &= q^{s(s+1)(2s+1)/12b}q^{-W_0/2b}
     \exp\left(\sum_{k=1}^\infty\Tbar_kJ_{-bk})\right)|s\rangle. 
\end{align*}
Consequently, the right half of (\ref{Z'ab(s,t,tbar)-fermion}) 
can be expressed as 
\begin{multline}
  e^{\Hbar(\bstbar)}\Gamma'_{-}(q^{-\rho})
  R_{b-1}^{L_0}\Gamma'_{-}(q^{-\rho})\cdots 
  R_1^{L_0}\Gamma'_{-}(q^{-\rho})|s\rangle\\
  = \exp\left(-\sum_{k=1}^\infty\frac{\tbar_k}{1-q^k}\right)
    q^{s(s+1)(2s+1)/12b}\prod_{1\le i\le j\le b-1}M(R_i\cdots R_j,q)^{-1}\\
  \mbox{}\times
    \Gamma'_{-}(q^{-\rho})\Gamma'_{+}(q^{-\rho})
    R_{b-1}^{L_0}\Gamma'_{-}(q^{-\rho})\Gamma'_{+}(q^{-\rho})\cdots\\
  \mbox{}\times
    R_1^{L_0}\Gamma'_{-}(q^{-\rho})\Gamma'_{+}(q^{-\rho})q^{-W_0/2b}
    \exp\left(\sum_{k=1}^\infty \Tbar_kJ_{-bk}\right)|s\rangle. 
\label{e^Hbar-to-right4}
\end{multline}

The inner product of (\ref{e^H-to-left4}) and 
(\ref{e^Hbar-to-right4}) yields the following expression 
of $Z'_{a,b}(s,\bst,\bstbar)$.  

\begin{prop}
\beq
  Z'_{a,b}(s,\bst,\bstbar) 
  = f'_{a,b}(s,\bst,\bstbar)\langle s|
    \exp\left(\sum_{k=1}^\infty T_kJ_{ak}\right)
    g'\exp\left(\sum_{k=1}^\infty\Tbar_kJ_{-bk}\right)|s\rangle, 
\label{Z'ab(s,t,tbar)-tau1}
\eeq
where
\begin{multline}
  f'_{a,b}(s,\bst,\bstbar) 
  = \exp\left(\sum_{k=1}^\infty\frac{t_kq^k-\tbar_k}{1-q^k}\right)
    q^{-s(s+1)(2s+1)(1/12a-1/12b)}\\
  \mbox{}\times
    \prod_{1\le i\le j\le a-1}M(P_i\cdots P_j,q)^{-1}\cdot
    \prod_{1\le i\le j\le b-1}M(R_i\cdots R_j,q)^{-1}  
\label{f'ab(s,t,tbar)}
\end{multline}
and 
\begin{multline}
  g' = q^{W_0/2a}\Gamma_{-}(q^{-\rho})\Gamma_{+}(q^{-\rho})
       P_1^{L_0}\Gamma_{-}(q^{-\rho})\Gamma_{+}(q^{-\rho})\cdots 
       P_{a-1}^{L_0}\Gamma_{-}(q^{-\rho})\Gamma_{+}(q^{-\rho})Q^{L_0}\\
   \mbox{}\times 
       \Gamma'_{-}(q^{-\rho})\Gamma'_{+}(q^{-\rho})
       R_{b-1}^{L_0}\Gamma'_{-}(q^{-\rho})\Gamma'_{+}(q^{-\rho})\cdots
       R_1^{L_0}\Gamma'_{-}(q^{-\rho})\Gamma'_{+}(q^{-\rho})q^{-W_0/2b}. 
\label{g'}
\end{multline}
\end{prop}

(\ref{Z'ab(s,t,tbar)-tau1}) shows that the partition function 
$Z'_{a,b}(s,\bst,\bstbar)$ is related to the tau function 
\beq
  \tau'(s,\bst,\bstbar) 
  = \langle s|\exp\left(\sum_{k=1}^\infty t_kJ_k\right)
    g'\exp\left(- \sum_{k=1}^\infty\tbar_kJ_{-k}\right)|s\rangle
\label{tau(g')}
\eeq
of the 2D Toda hierarchy restricted to a subspace 
of the full time evolutions.  Unlike the first orbifold model, 
the generating operator (\ref{g'}) seems to satisfy 
no particular algebraic relation like (\ref{Jg=gJ}). 

In summary, we have proved the following relation 
to the 2D Toda hierarchy.  

\begin{theorem}
The partition function $Z'_{a,b}(s,\bst,\bstbar)$ 
is related to the tau function (\ref{tau(g')}) 
of the 2D Toda hierarchy as 
\beq
  Z'_{a,b}(s,\bst,\bstbar) 
  = f'_{a,b}(s,\bst,\bstbar)\tau'(\bsT,-\bsTbar), 
\label{Z'ab(s,t,tbar)-tau2}
\eeq
where 
\beqnn
\begin{gathered}
  \bsT = (\underbrace{0,\ldots,0}_{a-1},T_1,
          \underbrace{0,\ldots,0}_{a-1},T_2,\ldots
          \underbrace{0,\ldots,0}_{a-1},T_k,\ldots),\\
  \bsTbar = (\underbrace{0,\cdots,0}_{b-1},\Tbar_1,
             \underbrace{0,\cdots,0}_{b-1},\Tbar_2,\ldots,
             \underbrace{0,\cdots,0}_{b-1},\Tbar_k,\ldots).  
\end{gathered}
\eeqnn
$T_k$'s and $\Tbar_k$'s are obtained from $t_k$'s and $\tbar_k$'s 
as shown in (\ref{T_k-t_k}) and (\ref{Tbar_k-tbar_k}). 
The prefactor $f'_{a,b}(s,\bst)$ is built from exponential 
and MacMahon functions as shown in (\ref{f'ab(s,t,tbar)}).  
\end{theorem}

\begin{remark}
The reduced time variables $\bsT$ and $\bsTbar$ 
in (\ref{Z'ab(s,t,tbar)-tau2}) have the same ``gaps'' structure 
as in the case of (\ref{Zab(s,t)-tau3}).  
These gaps disappear in the non-orbifold ($a = b = 1$) case. 
\end{remark}

\begin{remark}
The structure of the generating operator (\ref{g'}) is reminiscent 
of that of the generating operator for open string amplitudes 
on generalized conifolds of the bubble type \cite{Takasaki14}.  
An essential difference is the emergence 
of the fractional framing factors $q^{W_0/2a},q^{-W_0/2b}$   
in place of the ordinary ones $q^{W_0/2},q^{-W_0/2}$.  
\end{remark}

\section{Initial values of dressing operators}

\subsection{Matrix representation of generating operators}

The generating operators of Toda tau functions 
are non-degenerate elements of the Clifford group 
$\widehat{\GL}(\infty)$ \cite{MJD-book}.
\footnote{To give regular solutions of the Toda hierarchy, 
they have to satisfy the extra condition 
$\langle s|g|s\rangle \not= 0$ as well for all $s\in\ZZ$.} 
Such a Clifford operator is given by the exponential $e^{\hat{X}}$ 
of a fermion bilinear 
\beqnn
  \hat{X} = \sum_{m,n\in\ZZ}x_{mn}{:}\psi_{-m}\psi^*_n{:}  
\eeqnn
or a product of such exponential operators.  
Fermion bilinears of this form are in-one-to-one correspondence 
with $\ZZ\times\ZZ$ matrices 
\beqnn
  X = \sum_{m,n\in\ZZ}x_{mn}E_{mn},\quad 
  E_{mn} = (\delta_{im}\delta_{jn})_{i,j\in\ZZ}, 
\eeqnn
and satisfy the commutation relation 
\beqnn
  [\hat{X},\hat{Y}] = \widehat{[X,Y]} + c(X,Y), 
\eeqnn
where $c(X,Y)$ is a c-number cocycle of $\gl(\infty)$.  
These fermion bilinears thus form a central extension 
$\widehat{\gl}(\infty)$ of $\gl(\infty)$.  
Clifford operators $g$ correspond to 
$\ZZ\times\ZZ$ matrices $A = (a_{ij})_{i,j\in\ZZ}$ 
by the Bogoliubov transformation 
\beqnn
  g\psi_ng^{-1} = \sum_{m\in\in\ZZ}\psi_ma_{mn} 
\eeqnn
on the linear span of $\psi_n$'s.  
If $g$ is given by the exponential $e^{\hat{X}}$ 
of a fermion bilinear $\hat{X}$, $A$ becomes 
the exponential $e^X$ of the associated matrix $X$. 

Building blocks of the generating operators 
(\ref{g}) and (\ref{g'}) can be translated 
to $\ZZ\times\ZZ$ matrices as follows. 
In matrix representation (for which we use the same 
notations as those of fermion bilinears), 
the fundamental fermion bilinears $L_0,W_0,J_k$ 
can be expressed as 
\beq
  L_0 = \Delta,\quad 
  W_0 = \Delta^2,\quad 
  J_k = \Lambda^k, 
\eeq
where 
\beqnn
  \Delta = \sum_{n\in\ZZ}nE_{nn},\quad 
  \Lambda = \sum_{n\in\ZZ}E_{n,n+1}.
\eeqnn
Note that $\Lambda$ and $\Delta$ amount 
to the shift operator $e^{\rd_s}$ 
and the multiplication operator $s$ 
on the spatial lattice $\ZZ$ of the Toda hierarchy: 
\beq
  \Lambda \longleftrightarrow e^{\rd_s},\quad 
  \Delta \longleftrightarrow s. 
\label{matrix-diffop}
\eeq
$Q^{L_0}$, $P_i^{L_l}$, $R_j^{L_0}$, $q^{W_0/2a}$ 
and $q^{\pm W_0/2b}$ are Clifford operators, 
and their matrix representation can be obtained 
from the matrix representation of $L_0,W_0$ as 
\beq
\begin{gathered}
  Q^{L_0} = \sum_{n\in\ZZ}Q^nE_{nn},\quad 
  P_i^{L_0} = \sum_{n\in\ZZ}P_i^nE_{nn},\quad 
  R_j^{L_0} = \sum_{n\in\ZZ}R_j^nE_{nn},\\
  q^{W_0/2a} = \sum_{n\in\ZZ}q^{n^2/2a}E_{nn},\quad 
  q^{\pm W_0/2b} = \sum_{n\in\ZZ}q^{\pm n^2/2b}E_{nn}. 
\end{gathered}
\eeq
In the same way, matrix representation 
of the one-variable vertex operators 
$\Gamma_{\pm}(z)$ and $\Gamma_{\pm}(z)$ 
can be calculated as 
\beq
\begin{gathered}
  \Gamma_{\pm}(z) 
  = \exp\left(\sum_{k=1}^\infty\frac{z^k}{k}\Lambda^{\pm k}\right)
  = (1 - z\Lambda^{\pm 1})^{-1},\\
  \Gamma_{\pm}(z)
  = \exp\left(- \sum_{k=1}^\infty\frac{(-z)^k}{k}\Lambda^{\pm k}\right)
  = 1 + z\Lambda^{\pm 1}.
\end{gathered}
\eeq
Consequently, $\Gamma_{\pm}(q^{-\rho})$ and $\Gamma'_{\pm}(q^{-\rho})$ 
become infinite products of the form 
\beq
\begin{gathered}
  \Gamma_{\pm}(q^{-\rho}) 
  = \prod_{i=1}^\infty(1 - q^{i-1/2}\Lambda^{\pm 1})^{-1},\\
  \Gamma'_{\pm}(q^{-\rho})
  = \prod_{i=1}^\infty(1 + q^{i-1/2}\Lambda^{\pm 1}),
\end{gathered}
\label{matGamma-infprod}
\eeq
which may be thought of as matrix-valued quantum dilogarithm 
\cite{FV93,FK93}.  

Thus the generating operators (\ref{g}) and (\ref{g'}) 
correspond to the following matrices: 
\begin{multline}
  U = q^{\Delta^2/2a}
      \Gamma_{-}(q^{-\rho})\Gamma_{+}(q^{-\rho})P_1^\Delta
      \Gamma_{-}(q^{-\rho})\Gamma_{+}(q^{-\rho})\cdots 
      P_{a-1}^\Delta\Gamma_{-}(q^{-\rho})\Gamma_{+}(q^{-\rho})Q^\Delta\\
  \mbox{}\times 
      \Gamma_{-}(q^{-\rho})\Gamma_{+}(q^{-\rho})R_{b-1}^\Delta
      \Gamma_{-}(q^{-\rho})\Gamma_{+}(q^{-\rho})\cdots 
      R_1^\Delta\Gamma_{-}(q^{-\rho})\Gamma_{+}(q^{-\rho})q^{\Delta^2/2b}, 
\label{U}
\end{multline}
\begin{multline}
  U' = q^{\Delta^2/2a}
       \Gamma_{-}(q^{-\rho})\Gamma_{+}(q^{-\rho})P_1^\Delta
       \Gamma_{-}(q^{-\rho})\Gamma_{+}(q^{-\rho})\cdots 
       P_{a-1}^\Delta\Gamma_{-}(q^{-\rho})\Gamma_{+}(q^{-\rho})Q^\Delta\\
   \mbox{}\times 
       \Gamma'_{-}(q^{-\rho})\Gamma'_{+}(q^{-\rho})R_{b-1}^\Delta
       \Gamma'_{-}(q^{-\rho})\Gamma'_{+}(q^{-\rho})\cdots
       R_1^\Delta\Gamma'_{-}(q^{-\rho})\Gamma'_{+}(q^{-\rho})q^{-\Delta^2/2b}. 
\label{U'}
\end{multline}

\subsection{Matrix factorization problem}

Given a generating operator $g$ of the tau function, 
the associated solution of the 2D Toda hierarchy 
in the Lax formalism can be captured 
by a matrix factorization problem 
\cite{Takasaki84,NTT95,Takasaki95} of the form 
\beq
  \exp\left(\sum_{k=1}^\infty t_k\Lambda^k\right)
  U\exp\left(- \sum_{k=1}^\infty\tbar_k\Lambda^{-k}\right) 
  = W^{-1}\Wbar, 
\label{factor-problem}
\eeq
where $U$ is a matrix representation of $g$, 
$W = W(\bst,\bstbar)$ is a lower triangular matrix 
with all diagonal elements being equal to $1$, 
and $\Wbar = \Wbar(\bst,\bstbar)$ is an upper triangular matrix 
with all diagonal elements being non-zero, namely, 
\beqnn
\begin{gathered}
  W = 1 + w_1\Lambda^{-1} + w_2\Lambda^{-2} + \cdots,\\
  \Wbar = \wbar_0 + \wbar_1\Lambda + \wbar_2\Lambda^2 + \cdots,\\
\end{gathered}
\eeqnn
where $w_n$'s and $\wbar_n$'s are diagonal matrices 
and $\wbar_0$ is invertible.

$W$ and $\Wbar$ play the role of the dressing operators. 
They satisfy the Sato equations 
\beq
\begin{gathered}
  \frac{\rd W}{\rd t_k} 
     = - \left(W\Lambda^kW^{-1}\right)_{<0}W,\quad 
  \frac{\rd W}{\rd\bar{t}_k} 
     = \left(\bar{W}\Lambda^{-k}\bar{W}^{-1}\right)_{<0}W,\\
  \frac{\rd \bar{W}}{\rd t_k} 
     = \left(W\Lambda^kW^{-1}\right)_{\ge 0}\bar{W},\quad
  \frac{\rd\bar{W}}{\rd\bar{t}_k} 
     = - \left(\bar{W}\Lambda^{-k}\bar{W}^{-1}\right)_{\ge 0}\bar{W},
\end{gathered}
\label{Sato-eq}
\eeq
where $(\quad)_{\ge 0}$ and $(\quad)_{<0}$ denote 
the projection to the upper and strictly lower triangular parts, i.e., 
\beqnn
  \left(\sum_{m,n}a_{mn}E_{mn}\right)_{\ge 0} = \sum_{m\le n}a_{mn}E_{mn},\; 
  \left(\sum_{m,n}a_{mn}E_{mn}\right)_{<0} = \sum_{m>n}a_{mn}E_{mn}. 
\eeqnn
The Lax operators 
\beqnn
\begin{gathered}
  L = \Lambda + u_1 + u_2\Lambda^{-1} + \cdots,\\
  \Lbar^{-1} = \ubar_0\Lambda^{-1} + \ubar_1 + \ubar_2\Lambda + \cdots, 
\end{gathered}
\eeqnn
where $u_n$'s and $\ubar_n$'s are diagonal matrices, 
are obtained by ``dressing'' the shift matrices $\Lambda^{\pm 1}$ as 
\beq
  L = W\Lambda W^{-1}, \quad 
  \Lbar^{-1} = \bar{W}\Lambda^{-1}\bar{W}, 
\label{LLbar-WWbar}
\eeq
and satisfy the Lax equations 
\beq
\begin{gathered}
  \frac{\rd L}{\rd t_k} = [B_k,L], \quad
  \frac{\rd L}{\rd\bar{t}_k} = [\Bbar_k,L],\\
  \frac{\rd\Lbar^{-1}}{\rd t_k} = [B_k,\Lbar^{-1}],\quad
  \frac{\rd\Lbar^{-1}}{\rd\bar{t}_k} = [\Bbar_k,\Lbar^{-1}], 
\end{gathered}
\label{Lax-eq}
\eeq
where 
\beqnn
  B_k = (L^k)_{\ge 0}, \quad 
  \Bbar_k = (\bar{L}^{-k})_{<0}. 
\eeqnn

Solving the factorization problem (\ref{factor-problem}) 
for bi-infinite matrices directly 
is an extremely tough problem unlike its analogue 
for finite or semi-infinite matrices \cite{Takasaki84}. 
In a sense, the fermionic construction of tau functions 
is an alternative approach to this issue \cite{Takebe91}.  
In a cynical view, this approach simply converts 
a tough problem to another one, namely, to evaluating 
the right hand side of (\ref{tau(g)}).  

It is therefore remarkable that the factorization problem 
for the matrices (\ref{U}) and (\ref{U'}) can be solved 
explicitly at least {\it at the initial time\/}
$\bst = \bstbar = \bszero$.  

\begin{prop}
The generating matrices (\ref{U}) and (\ref{U'}) 
can be factorized as 
\beq
  U = W_{(0)}{}^{-1}\Wbar_{(0)},\quad 
  U' = W'_{(0)}{}^{-1}\Wbar'_{(0)}, 
\label{UU'-factorized}
\eeq
where 
\beq
\begin{aligned}
  W_{(0)} &= q^{\Delta^2/2a}\prod_{k=1}^{a+b}\Gamma_{-}(Q^{(k)}q^{-\rho})^{-1}
     \cdot q^{-\Delta^2/2a},\\
  \Wbar_{(0)} &= q^{\Delta^2/2a}\prod_{k=1}^{a+b}\Gamma_{+}(Q^{(k)-1}q^{-\rho})
    \cdot (P_1\cdots P_{a-1}QR_{b-1}\cdots R_1)^\Delta q^{\Delta^2/2b},
\end{aligned}
\label{WWbar(0)}
\eeq
and 
\beq
\begin{aligned}
  W'_{(0)} &= q^{\Delta^2/2a}\prod_{i=1}^a\Gamma_{-}(Q^{(i)}q^{-\rho})^{-1}
    \cdot\prod_{j=1}^b\Gamma'_{-}(Q^{(a+j)}q^{-\rho})^{-1}
    \cdot q^{-\Delta^2/2a},\\
  \Wbar'_{(0)} &= q^{\Delta^2/2a}\prod_{i=1}^a\Gamma_{+}(Q^{(i)-1}q^{-\rho})
    \cdot\prod_{j=1}^b\Gamma'_{+}(Q^{(a+j)-1}q^{-\rho})\\ 
    &\qquad\qquad \mbox{}\times 
    (P_1\cdots P_{a-1}QR_{b-1}\cdots R_1)^\Delta q^{-\Delta^2/2b}. 
\end{aligned}
\label{WWbar'(0)}
\eeq
The new constants $Q^{(i)}$'s are defined as 
\beq
\begin{gathered}
  Q^{(1)} = 1, \quad
  Q^{(i)} = P_1\cdots P_{i-1},\quad i=2,\ldots,a,\\
  Q^{(a+1)} = P_1\cdots P_{a-1}Q,\\
  Q^{(a+j)} = Q^{(a+1)}R_{b-1}\cdots R_{b-j+1},\quad j=2,\ldots,b.
\end{gathered}
\eeq
\end{prop}

\begin{proof}
One can use the matrix version 
\beq
\begin{gathered}
  \Gamma_{+}(vq^{-\rho})u^\Delta = u^\Delta\Gamma_{+}(uvq^{-\rho}),\quad 
  u^\Delta\Gamma_{-}(vq^{-\rho}) = \Gamma_{-}(uvq^{-\rho})u^\Delta,\\
  \Gamma'_{+}(vq^{-\rho})u^\Delta = u^{L_0}\Gamma'_{+}(uvq^{-\rho}),\quad
  u^\Delta\Gamma'_{-}(vq^{-\rho}) = \Gamma'_{-}(uvq^{-\rho})u^\Delta
\end{gathered}
\label{matGamma-scaling}
\eeq
of (\ref{Gamma-scaling}) and (\ref{Gamma'-scaling}) 
to collect $P_i^\Delta$'s, $R_j^\Delta$'s and $Q^\Delta$ 
in $U$ and $U'$ to the right end of the matrix product as 
\begin{multline*}
  U = q^{\Delta^2/2a}
    \prod_{i=1}^a\Gamma_{-}(Q^{(i)}q^{-\rho})\Gamma_{+}(Q^{(i)-1}q^{-\rho})\\
  \mbox{}\times
    \prod_{j=1}^b\Gamma_{-}(Q^{(a+j)}q^{-\rho})\Gamma_{+}(Q^{(a+j)-1}q^{-\rho})
    \cdot (P_1\cdots P_{a-1}QR_{b-1}\cdots R_1)^\Delta q^{\Delta^2/2b}
\end{multline*}
and 
\begin{multline*}
  U' = q^{\Delta^2/2a}
     \prod_{i=1}^a\Gamma_{-}(Q^{(i)}q^{-\rho})\Gamma_{+}(Q^{(i)-1}q^{-\rho})\\
   \mbox{}\times
     \prod_{j=1}^b\Gamma'_{-}(Q^{(a+j)}q^{-\rho})\Gamma'_{+}(Q^{(a+j)-1}q^{-\rho})
     \cdot (P_1\cdots P_{a-1}QR_{b-1}\cdots R_1)^\Delta q^{-\Delta^2/2b}. 
\end{multline*}
Since $\Gamma_{\pm}$'s and $\Gamma'_{\pm}$'s are commutative, 
one can move $\Gamma_{-}$'s to the left side, 
$\Gamma_{+}$'s and $\Gamma'_{+}$'s to the right side, 
and insert $1 = q^{-\Delta^2/2a}\cdot q^{\Delta^2/2a}$ in the middle 
to achieve the factorization (\ref{UU'-factorized}). 
\end{proof}

This result implies that $W_{(0)},\Wbar_{(0)}$ and 
$W'_{(0)},\Wbar'_{(0)}$ are nothing  but the initial values 
$W|_{\bst=\bstbar=\bszero},\Wbar|_{\bst=\bstbar=\bszero}$ 
of the dressing operators for the special solutions 
of the 2D Toda hierarchy defined by the tau functions 
(\ref{tau(g)}) and (\ref{tau(g')}).  
These initial values, in turn, determine the initial values 
$L|_{\bst=\bstbar=\bszero},\Lbar|_{\bst=\bstbar=\bszero}$ 
of the Lax operators by the dressing relation (\ref{LLbar-WWbar}). 
This eventually leads to a precise characterization 
of these special solutions as we shall show in the next section.

\section{Structure of Lax operators}

\subsection{Initial values of Lax operators: First orbifold model}

Let $L_{(0)}$ and $\Lbar_{(0)}{}^{-1}$ denote the initial values 
of the Lax operators determined by $U$.  
As it turns out below, it is $L_{(0)}{}^a$ and $\Lbar_{(0)}{}^{-b}$ 
rather than $L_{(0)}$ and $\Lbar_{(0)}{}^{-1}$ themselves 
that play a fundamental role in the subsequent consideration.  
One can find their explicit form from (\ref{WWbar(0)}). 

The following are technical clues for these calculations. 

\begin{lemma}
\beq
\begin{gathered}
  q^{-\Delta^2/2a}\Lambda^a q^{\Delta^2/2a} = q^{a/2}q^\Delta \Lambda^a,\\
  u^\Delta q^{\Delta^2/2b}\Lambda^{-b}q^{-\Delta^2/2b}u^{-\Delta} 
    = u^bq^{-b/2}q^\Delta\Lambda^{-b}. 
\end{gathered}
\label{Lax(0)key1}
\eeq
\end{lemma}

\begin{proof}
Do straightforward calculations. 
\end{proof}

\begin{lemma}
\beq
\begin{gathered}
  \Gamma_{-}(uq^{-\rho})^{-1}q^\Delta\Gamma_{-}(uq^{-\rho}) 
  = q^\Delta(1 - uq^{-1/2}\Lambda^{-1}),\\
  \Gamma_{+}(u^{-1}q^{-\rho})q^\Delta\Gamma_{+}(u^{-1}q^{-\rho})^{-1}
  = q^\Delta(1 - u^{-1}q^{1/2}\Lambda).
\end{gathered}
\label{Lax(0)key2}
\eeq
\end{lemma}

\begin{proof}
To derive the first identity, one can use 
the scaling property (\ref{matGamma-scaling}) 
and the infinite product form (\ref{matGamma-infprod}) 
to rewrite the left hand side as 
\begin{align*}
  \Gamma_{-}(uq^{-\rho})^{-1}q^\Delta\Gamma_{-}(uq^{-\rho}) 
  &= q^\Delta\Gamma_{-}(uq^{-1}q^{-\rho})^{-1}\Gamma_{-}(uq^{-\rho})\nonumber\\
  &= q^\Delta\prod_{i=1}^\infty(1 - uq^{i-3/2}\Lambda^{-1})
     \cdot\prod_{i=1}^\infty(1 - uq^{i-1/2}\Lambda^{-1})^{-1}\nonumber\\
  &= q^\Delta(1 - uq^{-1/2}\Lambda^{-1}). 
\end{align*}
In the same way, the second identity can be derived as 
\begin{align*}
  \Gamma_{+}(u^{-1}q^{-\rho})q^\Delta\Gamma_{+}(u^{-1}q^{-\rho})^{-1}
  &= q^\Delta\Gamma_{+}(u^{-1}qq^{-\rho})\Gamma_{+}(u^{-1}q^{-\rho})^{-1}\nonumber\\
  &= q^\Delta\prod_{i=1}^\infty(1 - u^{-1}q^{i+1/2}\Lambda)^{-1}
     \cdot\prod_{i=1}^\infty(1 - u^{-1}q^{i-1/2}\Lambda)\nonumber\\
  &= q^\Delta(1 - u^{-1}q^{1/2}\Lambda). 
\end{align*}
\end{proof}

Plugging (\ref{WWbar(0)}) into the dressing relations 
\beqnn
  L_{(0)}{}^a = W_{(0)}\Lambda^a W_{(0)}{}^{-1},\quad 
  \Lbar_{(0)}{}^{-b} =\Wbar_{(0)}\Lambda^{-b}\Wbar_{(0)}{}^{-1}
\eeqnn
and applying (\ref{Lax(0)key1}), one can express 
$L_{(0)}{}^a$ and $\Lbar_{(0)}{}^{-b}$ as 
\begin{align*}
  L_{(0)}{}^a 
  &= q^{a/2}q^{\Delta^2/2a}\prod_{k=1}^{a+b}\Gamma_{-}(Q^{(k)}q^{-\rho})^{-1}
    \cdot q^\Delta\Lambda^a\prod_{k=1}^{a+b}\Gamma_{-}(Q^{(k)}q^{-\rho})
    \cdot q^{-\Delta^2/2a}\nonumber\\
  &= q^{a/2}q^{\Delta^2/2a}\prod_{k=1}^{a+b}\Gamma_{-}(Q^{(k)}q^{-\rho})^{-1}
    \cdot q^\Delta\prod_{k=1}^{a+b}\Gamma_{-}(Q^{(k)}q^{-\rho})
    \cdot \Lambda^a q^{-\Delta^2/2a}
\end{align*}
and 
\begin{align*}
  \Lbar_{(0)}{}^{-b}
  &= (P_1\cdots P_{a-1}QR_{b-1}\cdots R_1)^b\nonumber\\
  &\quad\mbox{}\times
    q^{-b/2}q^{\Delta^2/2a}\prod_{k=1}^{a+b}\Gamma_{+}(Q^{(k)-1}q^{-\rho})
    \cdot q^\Delta\Lambda^{-b}
    \prod_{k=1}^{a+b}\Gamma_{+}(Q^{(k)-1}q^{-\rho})^{-1}
    \cdot q^{-\Delta^2/2a}\nonumber\\
  &= (P_1\cdots P_{a-1}QR_{b-1}\cdots R_1)^b\nonumber\\
  &\quad\mbox{}\times
    q^{-b/2}q^{\Delta^2/2a}\prod_{k=1}^{a+b}\Gamma_{+}(Q^{(k)-1}q^{-\rho})
    \cdot q^\Delta\prod_{k=1}^{a+b}\Gamma_{+}(Q^{(k)-1}q^{-\rho})^{-1}
    \cdot \Lambda^{-b}q^{-\Delta^2/2a}. 
\end{align*}
By repeated use of (\ref{Lax(0)key2}), 
one can rewrite these expressions as 
\beq
  L_{(0)}{}^a 
  = q^{a/2}q^{\Delta^2/2a}q^\Delta
    \prod_{k=1}^{a+b}(1 - Q^{(k)}q^{-1/2}\Lambda^{-1})
    \cdot \Lambda^a q^{-\Delta^2/2a}
\label{L(0)-factorized}
\eeq
and 
\begin{align}
  \Lbar_{(0)}{}^{-b} 
  &= (P_1\cdots P_{a-1}QR_{b-1}\cdots R_1)^b\nonumber\\
  &\quad\mbox{}\times
    q^{-b/2}q^{\Delta^2/2a}q^\Delta
    \prod_{k=1}^{a+b}(1 - Q^{(k)-1}q^{1/2}\Lambda)
    \cdot \Lambda^{-b}q^{-\Delta^2/2a}. 
\label{Lbar(0)-factorized}
\end{align}

Because of the identity 
\beqnn
  \prod_{k=1}^{a+b}(1 - Q^{(k)-1}q^{1/2}\Lambda)\cdot\Lambda^{-b}
  = \prod_{k=1}^{a+b}(-Q^{(k)-1}q^{1/2})
    \cdot\prod_{k=1}^{a+b}(1 - Q^{(k)}q^{-1/2}\Lambda^{-1})
    \cdot\Lambda^a, 
\eeqnn
$L_{(0)}{}^a$ and $\Lbar_{(0)}{}^{-b}$ turn out to coincide 
with each other up to a constant factor: 
\beqnn
  \Lbar_{(0)}{}^{-b}
  = (P_1\cdots P_{a-1}QR_{b-1}\cdots R_1)^b
    \prod_{k=1}^{a+b}(-Q^{(k)-1})
    \cdot L_{(0)}{}^a.
\eeqnn
Moreover, since 
\begin{multline*}
   \prod_{k=1}^{a+b}(1 - Q^{(k)}q^{-1/2}\Lambda^{-1})\cdot \Lambda^a 
  = \prod_{i=1}^a(\Lambda - Q^{(i)}q^{-1/2})
    \cdot\prod_{j=1}^b(1 - Q^{(a+j)}q^{-1/2}\Lambda^{-1}),
\end{multline*}
the foregoing results (\ref{L(0)-factorized}) 
and (\ref{Lbar(0)-factorized}) can be restated 
in the following form. 

\begin{prop}
$L_{(0)}^a$ and $\Lbar_{(0)}^{-b}$ can be factorized as 
\beq
  L_{(0)}{}^a = D^{-1}\Lbar_{(0)}{}^{-b} = B_{(0)}C_{(0)}, 
\label{LLbar(0)-BC(0)}
\eeq
where 
\begin{gather}
  B_{(0)} = q^{a/2}q^{\Delta^2/2a}q^\Delta
      \prod_{i=1}^a(\Lambda - Q^{(i)}q^{-1/2})\cdot q^{-\Delta^2/2a},\\
  C_{(0)} = q^{\Delta^2/2a}\prod_{j=1}^b(1 - Q^{(a+j)}q^{-1/2}\Lambda^{-1})
      \cdot q^{-\Delta^2/2a},\\
  D = (P_1\cdots P_{a-1}QR_{b-1}\cdots R_1)^b\prod_{k=1}^{a+b}(-Q^{(k)-1}). 
\end{gather}
$B_{(0)}$ and $C_{(0)}$ are polynomials in $\Lambda^{\pm 1}$ 
of the form 
\beq
\begin{gathered}
  B_{(0)} = \Lambda^a + \beta_{1(0)}\Lambda^{a-1} 
    + \cdots \beta_{a(0)},\\
  C_{(0)} = 1 + \gamma_{1(0)}\Lambda^{-1} 
    + \cdots + \gamma_{b(0)}\Lambda^{-b},
\end{gathered}
\eeq
where $\beta_{i(0)}$'s and $\gamma_{j(0)}$'s are diagonal matrices.  
\end{prop}

\subsection{Initial values of Lax operators: Second orbifold model}

Let $L'_{(0)}$ and $\Lbar'_{(0)}{}^{-1}$ denote 
the initial values of the Lax operators determined by $U'$. 
In this case, too, $L_{(0)}{}^a$ and $\Lbar_{(0)}{}^{-b}$ 
play a fundamental role.  Let us derive their explicit form 
from (\ref{WWbar'(0)}). 

Technical clues are the identities (\ref{Lax(0)key1}), 
(\ref{Lax(0)key2}) and their variants 
\begin{gather}
  u^\Delta q^{-\Delta^2/2b}\Lambda^{-b}q^{\Delta^2/2b}u^{-\Delta} 
  = u^bq^{b/2}q^{-\Delta}\Lambda^{-b}, 
  \label{Lax(0)key3}\\
  \Gamma'_{-}(uq^{-\rho})^{-1}q^\Delta\Gamma'_{-}(uq^{-\rho})
    = q^\Delta(1 + uq^{-1/2}\Lambda^{-1})^{-1} 
  \label{Lax(0)key4}
\end{gather}
and 
\beq
\begin{gathered}
  \Gamma_{+}(u^{-1}q^{-\rho})q^{-\Delta}\Gamma_{+}(u^{-1}q^{-\rho})^{-1}
    = q^{-\Delta}(1 - u^{-1}q^{-1/2}\Lambda)^{-1},\\
  \Gamma'_{+}(u^{-1}q^{-\rho})q^{-\Delta}\Gamma'_{+}(u^{-1}q^{-\rho})^{-1}
    = q^{-\Delta}(1 + u^{-1}q^{-1/2}\Lambda). 
\end{gathered}
\label{Lax(0)key5}
\eeq
By the identities (\ref{Lax(0)key1}) and (\ref{Lax(0)key3})  
and the dressing relations 
\beqnn
  L'_{(0)}{}^a = W'_{(0)}\Lambda^a W'_{(0)}{}^{-1},\quad 
  \Lbar'_{(0)}{}^{-b} =\Wbar'_{(0)}\Lambda^{-b}\Wbar'_{(0)}{}^{-1}, 
\eeqnn
$L'_{(0)}{}^a$ and $\Lbar'_{(0)}{}^{-b}$ can be expressed as 
\begin{align*}
  L'_{(0)}{}^a 
  &= q^{a/2}q^{\Delta^2/2a}\prod_{i=1}^a\Gamma_{-}(Q^{(i)}q^{-\rho})^{-1}
     \cdot\prod_{j=1}^b\Gamma'_{-}(Q^{(a+j)}q^{-\rho})^{-1}\nonumber\\
  &\quad\mbox{}\times
     q^\Delta\prod_{j=1}^b\Gamma'_{-}(Q^{(a+j)}q^{-\rho})
     \cdot\prod_{i=1}^a\Gamma_{-}(Q^{(i)}q^{-\rho})
     \cdot\Lambda^aq^{-\Delta^2/2a}
\end{align*}
and 
\begin{align*}
  \Lbar'_{(0)}{}^{-b}
  &= (P_1\cdots P_{a-1}QR_{b-1}\cdots R_1)^b\nonumber\\
  &\quad\mbox{}\times
     q^{b/2}q^{\Delta^2/2a}\prod_{i=1}^a\Gamma_{+}(Q^{(i)-1}q^{-\rho})
     \cdot\prod_{j=1}^b\Gamma'_{+}(Q^{(a+j)-1}q^{-\rho})\nonumber\\
  &\quad\mbox{}\times
     q^{-\Delta}\prod_{j=1}^b\Gamma'_{+}(Q^{(a+j)-1}q^{-\rho})^{-1}
     \cdot\prod_{i=1}^a\Gamma_{+}(Q^{(i)-1}q^{-\rho})^{-1}
     \cdot\Lambda^{-b}q^{\Delta^2/2a}. 
\end{align*}
By (\ref{Lax(0)key2}), (\ref{Lax(0)key4}) and (\ref{Lax(0)key5}), 
the right hand side can be simplified as 
\begin{align}
  L'_{(0)}{}^a 
  &= q^{a/2}q^{\Delta^2/2a}q^\Delta 
    \prod_{i=1}^a(1 - Q^{(i)}q^{-1/2}\Lambda^{-1})\nonumber\\
  &\quad\mbox{}\times 
     \prod_{j=1}^b(1 + Q^{(a+j)}q^{-1/2}\Lambda^{-1})^{-1}
     \cdot\Lambda^a q^{-\Delta^2/2a}
\label{L'(0)-factorized}
\end{align}
and 
\begin{align}
  \Lbar'_{(0)}{}^{-b} 
  &= (P_1\cdots P_{a-1}QR_{b-1}\cdots R_1)^b\nonumber\\
  &\quad\mbox{}\times
     q^{b/2}q^{\Delta^2/2a}q^{-\Delta}
     \prod_{i=1}^a(1 - Q^{(i)-1}q^{-1/2}\Lambda)^{-1}\nonumber\\
  &\quad\mbox{}\times
     \prod_{j=1}^b(1 + Q^{(a+j)-1}q^{-1/2}\Lambda)
     \cdot\Lambda^{-b}q^{-\Delta^2/2a}. 
\label{Lbar'(0)-factorized}
\end{align}

Since 
\begin{align*}
  &q^{a/2}q^\Delta\prod_{i=1}^a(1 - Q^{(i)}q^{-1/2}\Lambda^{-1})
   \cdot\prod_{j=1}^b(1 + Q^{(a+j)}q^{-1/2}\Lambda^{-1})^{-1}\cdot\Lambda^a\\
  &= q^{a/2}q^\Delta\prod_{i=1}^a(\Lambda - Q^{(i)}q^{-1/2})
    \cdot\prod_{j=1}^b(1 + Q^{(a+j)}q^{-1/2}\Lambda^{-1})^{-1} 
\end{align*}
and 
\begin{align*}
  &q^{b/2}q^{-\Delta}\prod_{i=1}^a(1 - Q^{(i)-1}q^{-1/2}\Lambda)^{-1}
   \cdot\prod_{j=1}^b(1 + Q^{(a+j)-1}q^{-1/2}\Lambda)\cdot\Lambda^{-b}\\
  &= q^{b/2}\prod_{i=1}^a(1 - Q^{(i)-1}q^{1/2}\Lambda)^{-1}
     \cdot\prod_{j=1}^b(1 + Q^{(a+j)-1}q^{1/2}\Lambda)
     \cdot \Lambda^{-b}q^{-b}q^{-\Delta}\\
  &= \prod_{i=1}^a(-Q^{(i)})\cdot\prod_{j=1}^bQ^{(a+j)-1}
     \cdot\prod_{j=1}^b(1 + Q^{(a+j)}q^{-1/2}\Lambda^{-1})\\
  &\qquad\mbox{}\times 
     \left(q^{a/2}q^\Delta\prod_{i=1}^a(\Lambda - Q^{(i)}q^{-1/2})\right)^{-1}, 
\end{align*}
the foregoing results (\ref{L'(0)-factorized}) 
and (\ref{Lbar'(0)-factorized}) can be restated as follows. 

\begin{prop}
$L'_{(0)}{}^a$ and $\Lbar'_{(0)}{}^{-b}$ can be factorized as 
\beq
  L'_{(0)}{}^a = B'_{(0)}C'_{(0)}{}^{-1},\quad 
  \Lbar'_{(0)}{}^{-b} = D'C'_{(0)}B'_{(0)}{}^{-1}, 
\label{LLbar'(0)-BC'(0)}
\eeq
where 
\begin{gather}
  B'_{(0)} = q^{a/2}q^{\Delta^2/2a}q^\Delta
      \prod_{i=1}^a(\Lambda - Q^{(i)}q^{-1/2})\cdot q^{-\Delta^2/2a},\\
  C'_{(0)} = q^{\Delta^2/2a}\prod_{j=1}^b(1 + Q^{(a+j)}q^{-1/2}\Lambda^{-1})
      \cdot q^{-\Delta^2/2a},\\
  D' = (P_1\cdots P_{a-1}QR_{b-1}\cdots R_1)^b 
       \prod_{i=1}^a(-Q^{(i)})\cdot\prod_{j=1}^bQ^{(a+j)-1}. 
\end{gather}
$B'_{(0)}$ and $C'_{(0)}$ are polynomials in $\Lambda^{\pm 1}$ 
of the form 
\beq
\begin{gathered}
  B'_{(0)} = \Lambda^a + \beta'_{1(0)}\Lambda^{a-1} 
    + \cdots \beta'_{a(0)},\\
  C'_{(0)} = 1 + \gamma'_{1(0)}\Lambda^{-1} 
    + \cdots + \gamma'_{b(0)}\Lambda^{-b},
\end{gathered}
\eeq
where $\beta'_{i(0)}$'s and $\gamma'_{j(0)}$'s are diagonal matrices.  
The inverse matrices of $B'_{(0)}$ and $C'_{(0)}$ are understood 
to be power series of $\Lambda^{\pm 1}$ of the form 
\beq
\begin{gathered}
  B'_{(0)}{}^{-1} = \beta'_{a(0)}{}^{-1} 
    - \beta'_{a-1(0)}\beta'_{a(0)}{}^{-1}\Lambda\beta'_{a(0)}{}^{-1} 
    + \cdots,\\
  C'_{(0)}{}^{-1} = 1 - \gamma'_{1(0)}\Lambda^{-1} + \cdots. 
\end{gathered}
\eeq
\end{prop}

\subsection{Reductions of 2D Toda hierarchy}

The factorized expressions (\ref{LLbar(0)-BC(0)}) 
and (\ref{LLbar'(0)-BC'(0)}) of the initial values 
of the Lax operators imply that these special solutions 
of the 2D Toda hierarchy belong to the following reductions. 

\begin{itemize}
\item[(i)] 
{\it Bi-graded Toda hierarchy of bi-degree $(a,b)$\/} 
\cite{Kuperschmidt85,Carlet06}: 
This reduction is characterized by the algebraic relation 
\beq
  L^a = D^{-1}\Lbar^{-b},
\label{L^a=Lbar^(-b)}
\eeq
where $D$ is a non-zero constant.  
$D$ can be normalized to $D = 1$ by rescaling $\tbar_k$'s.  
Under the reduction condition (\ref{L^a=Lbar^(-b)}), 
both sides become a Laurent polynomial 
\beq
  \frakL = \Lambda^a + \alpha_1\Lambda^{a-1} + \cdots 
           + \alpha_{a+b}\Lambda^{-b} 
\eeq
of bi-degree $(a,b)$ in $\Lambda$ 
with $\alpha_k$'s being diagonal matrices. 
The Lax equations (\ref{Lax-eq}) of the 2D Toda hierarchy 
thereby reduce to the Lax equations 
\beq
  \frac{\rd\frakL}{\rd t_k} = [B_k,\frakL],\quad 
  \frac{\rd\frakL}{\rd\tbar_k} = [\Bbar_k,\frakL]
\eeq
for the reduced Lax matrix $\frakL$. 
$L$ and $\Lbar^{-1}$ can be reconstructed 
from $\frakL$ as fractional powers 
\beqnn
\begin{gathered}
  \frakL^{1/a} = \Lambda + u_1 + u_2\Lambda^{-1} + \cdots,\\
  \frakL^{1/b} = \ubar_0\Lambda^{-1} + \ubar_1 + \ubar_2\Lambda + \cdots
\end{gathered}
\eeqnn
of two different (descending and ascending) types.  
A special case of this reduction is obtained 
by assuming the factorization 
\beq
  L^a = D^{-1}\Lbar^{-b} = BC 
\label{LLbar-BC(i)}
\eeq
of $L^a$ and $\Lbar^{-b}$, where $B$ and $C$ 
are polynomials in $\Lambda^{\pm 1}$ of the form 
\beq
\begin{gathered}
  B = \Lambda^a + \beta_1\Lambda^{a-1} + \beta_a,\\
  C = 1 + \gamma_1\Lambda^{-1} + \cdots + \gamma_b\Lambda^{-b} 
\end{gathered}
\label{BC}
\eeq
with $\beta_i$'s and $\gamma_j$'s being diagonal matrices.  
It is easy to see that the factorized form  
(\ref{LLbar-BC(i)}) of $L^a$ and $\Lbar^{-b}$ 
is preserved by the time evolutions of the 2D Toda hierarchy.  
The proof is fully parallel to the case of 
the rational reductions \cite{BCRR14}.  
Thus the initial values $L_{(0)}$ and $\Lbar_{(0)}{}^{-1}$ 
determine a special solution of the bi-graded Toda hierarchy 
of bi-degree $(a,b)$.  

\item[(ii)] 
{\it Rational reduction of bi-degree $(a,b)$\/} \cite{BCRR14}: 
This reduction is characterized by the factorization 
\beq
  L^a = BC^{-1}, \quad \Lbar^{-b} = DCB^{-1}, 
\label{LLbar-BC(ii)}
\eeq
where $B$ and $C$ are matrices of the same form as (\ref{BC}) 
and $D$ is a non-zero constant.  $D$ can be normalized 
to $D = 1$ by rescaling $\tbar_k$'s.  
The inverse matrices $B^{-1}$ and $C^{-1}$ are understood 
to be power series of $\Lambda^{\pm 1}$ of the form 
\beq
\begin{gathered}
  B^{-1} = \beta_a^{-1} - \beta_{a-1}\beta_a^{-1}\Lambda\beta_a^{-1}
    + \cdots,\\
  C^{-1} = 1 - \gamma_1\Lambda^{-1} + \cdots. 
\end{gathered}
\eeq
$L$ and $\Lbar^{-1}$ can be reconstructed from $B$ and $C$ as 
\beqnn
\begin{gathered}
  L = (BC^{-1})^{1/a} = \Lambda + u_1 + u_2\Lambda^{-1} + \cdots,\\
  \Lbar^{-1} = (DCB^{-1})^{1/b} = \ubar_0\Lambda^{-1} + \ubar_1 
    + \ubar_2\Lambda + \cdots. 
\end{gathered}
\eeqnn
As shown by Brini et al. \cite{BCRR14}, the factorized form 
(\ref{LLbar-BC(ii)}) of $L^a$ and $\Lbar^{-b}$ is preserved 
by the time evolutions of the 2D Toda hierarchy.  
The reduced integrable hierarchy is a generalization 
of the Ablowitz-Ladik hierarchy \cite{AL75,Vekslerchik97} 
or, equivalently, the relativistic Toda hierarchy 
\cite{Ruijsenaars90,KMZ96,Suris97}, which amounts 
to the case where $a = b = 1$ \cite{BCR11}. 
Thus the initial values $L'_{(0)}$ and $\Lbar'_{(0)}{}^{-1}$ 
determines a special solution of the rational reduction 
of bi-degree $(a,b)$.  
\end{itemize}

We have thus arrived at the following conclusion. 

\begin{theorem}
\quad
\begin{itemize}
\item[\rm (i)] The generating matrix (\ref{U}) 
determines a special solution of the 2D Toda hierarchy 
for which the Lax operators have the factorized form 
(\ref{LLbar-BC(i)}) at all time.  The first orbifold model 
thus corresponds to a special solution 
of the bi-graded Toda hierarchy of bi-degree $(a,b)$.  
\item[\rm (ii)] The generating matrix (\ref{U'}) 
determines a special solution of the 2D Toda hierarchy 
for which the Lax operators have the factorized form 
(\ref{LLbar-BC(ii)}) at all time.  The second orbifold model 
thus corresponds to a special solution 
of the rational reduction of bi-degree $(a,b)$.  
\end{itemize}
\end{theorem}

\begin{remark}
In the case of open string amplitudes 
on a generalized conifold of the bubble type \cite{Takasaki14}, 
the Lax operators take yet another factorized form 
\beq
  L = B\Lambda^{1-N}C^{-1}, \quad 
  \Lbar^{-1} = DC\Lambda^{N-1}B^{-1}, 
\label{LLbar-BC(iii)}
\eeq
where and $B$ and $C$ are polynomials in $\Lambda^{\pm 1}$ 
of the form 
\beq
\begin{gathered}
  B = \Lambda^N + \beta_1\Lambda^{N-1} + \cdots + \beta_N,\\
  C = 1 + \gamma_1\Lambda^{-1} + \cdots + \gamma_N\Lambda^{-N} 
\end{gathered}
\eeq
with $\beta_i$'s and $\gamma_i$'s being diagonal matrices.  
An obvious difference between (\ref{LLbar-BC(ii)}) 
and (\ref{LLbar-BC(iii)}) is that the former, unlike the latter, 
contains the powers $L^a,\Lbar^{-b}$ of the Lax operators. 
This difference stems from the difference of the framing factors, 
namely, $q^{\Delta^2/2a},q^{\Delta^2/2b}$ in the orbifold case 
and $q^{\Delta^2/2}$ in the ordinary ($a = b = 1$) case.  
The fractional framing factors $q^{\Delta^2/2a},q^{\Delta^2/2b}$ 
yield the terms $\Lambda^a,\Lambda^{-b}$ 
in (\ref{Lax(0)key1}) and (\ref{Lax(0)key3}).  
These powers of $\Lambda^{\pm 1}$ eventually turn into 
the powers $L^a,\Lbar^{-b}$ of $L,\Lbar^{-1}$. 
\end{remark}

\subsection*{Acknowledgements}

This work is partly supported by JSPS KAKENHI Grant 
No. 24540223 and No. 25400111.

\end{document}